\documentclass[preprint]{revtex4-1}
\usepackage{amsmath}
\usepackage{amsthm}
\usepackage{tabularx}
\usepackage{multirow}
\usepackage{centernot}
\usepackage{hyperref}
\usepackage{graphicx}
\usepackage{booktabs}
\usepackage{calc}
\usepackage{cancel} 
\begin{document}

\newcommand{\odd}{\mathrm{odd}}

\newcommand{\vct}[1]{\mathbf{#1}}

\newcolumntype{C}{>{\centering\arraybackslash}X}
\newcolumntype{L}{>{\raggedright\arraybackslash}X}
\newcolumntype{R}{>{\raggedleft\arraybackslash}X}

\newcommand{\vx}{\vct x}
\newcommand{\Bset}{\mathcal B}
\newcommand{\NB}{N_\Bset}
\newcommand{\RX}{(R, \lambda)}

\newcommand{\eq}{Eq.}
\newcommand{\eqs}{Eqs.}
\newcommand{\req}[1]{(\ref{eq:#1})}
\newcommand{\refeq}[1]{\eq\,\req{#1}}
\newcommand{\refeqs}[1]{\eqs\,\req{#1}}
\newcommand{\reqsub}[2]{(\ref{eq:#1}#2)}
\newcommand{\refeqsub}[2]{\eq\,\reqsub{#1}{#2}}
\newcommand{\refeqssub}[2]{\eqs\,\reqsub{#1}{#2}}

\newcommand{\refthm}[1]{Theorem \ref{thm:#1}}
\newcommand{\refthms}[1]{Theorems \ref{thm:#1}}
\newcommand{\refsec}[1]{Section \ref{sec:#1}}
\newcommand{\refsecs}[1]{Sections \ref{sec:#1}}
\newcommand{\refapd}[1]{Appendix \ref{apd:#1}}
\newcommand{\reftab}[1]{Table \ref{tab:#1}}
\newcommand{\reftabs}[1]{Tables \ref{tab:#1}}
\newcommand{\reffig}[1]{Fig. \ref{fig:#1}}
\newcommand{\reffigs}[1]{Figs. \ref{fig:#1}}

\newtheorem{theorem}{Theorem}
\newenvironment{remark}[1][1]%
{\par\noindent\textbf{Remark #1.} }{\medskip}

\title{Cycles of the logistic map}
\author{Cheng Zhang}
\affiliation{Applied Physics Program and Department of Bioengineering, \
Rice Univeristy, Houston, TX 77005}

\begin{abstract}
%
The onset and bifurcation points of the $n$-cycles of
  a polynomial map are located
  through a characteristic equation
  connecting cyclic polynomials formed by periodic orbit points.
The minimal polynomials of the critical parameters
  of the logistic, H\'enon, and cubic maps are obtained 
  for $n$ up to 13, 9, and 8,
  respectively.
%
\end{abstract}
\maketitle

\section{Introduction}

Consider the logistic map \cite{may, strogatz}:
\begin{equation}
  x_{k+1} = f(x_k) \equiv r \, x_k \, ( 1 - x_k ).
\label{eq:logmap}
\end{equation}
If we iterate \refeq{logmap} from $k = 0$,
what does the resulting sequence $x_0$,
  $x_1 = f(x_0)$,
  $x_2 = f(x_1)$, $\ldots$ look like?
We can visualize the sequence on the cobweb plot,
  see \reffig{cobweb} for examples.
Starting from $(x_0, x_0)$ on the diagonal,
  each vertical arrow takes $(x_k, x_k)$ to $(x_k, y)$,
  where $y = f(x_k) = x_{k+1}$;
the next horizontal arrow then reflects $(x_k, y)$ to
  $(y, y) = (x_{k+1}, x_{k+1})$,
  which starts the next iteration.

\begin{figure}[h]
  \begin{minipage}{\linewidth}
    \begin{center}
        \includegraphics[angle=-90, width=\linewidth]{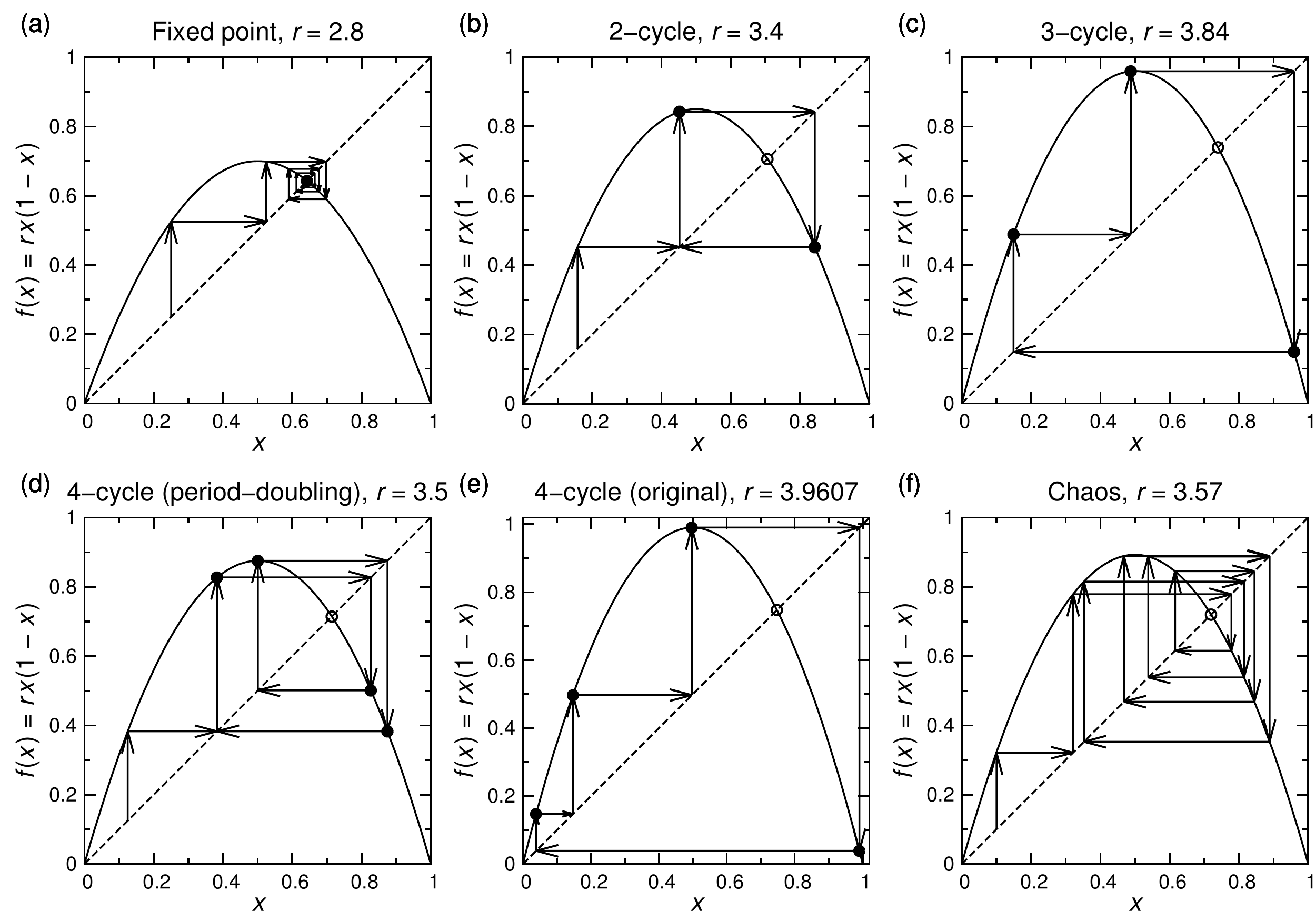}
    \end{center}
  \end{minipage}%
  \caption{\label{fig:cobweb}
  Cobweb plot of the logistic map.}
\end{figure}

Three outcomes are possible: 
  (i) a \emph{fixed point}, 
      which is a constant (including infinity),
      e.g., \reffig{cobweb}(a);
  (ii) a periodic \emph{cycle},
      which is a self-repeating pattern,
      e.g., \reffigs{cobweb}(b)-(e);
or
  (iii) a chaotic trajectory,
      e.g., \reffig{cobweb}(f).
We will focus on the first two cases here.

A fixed point is a solution of $x^* = f(x^*)$.
If $x_0$ deviates slightly from $x^*$
  and the sequence still converges to $x^*$,
  we call it \emph{stable}.
For a differentiable $f$,
  a stable fixed point requires $|f'(x^*)| \le 1$
  to reduce deviations in successive iterations \cite{strogatz}.

In an $n$-cycle, $n$ is the smallest positive integer 
  that allows $x_1 = x_{n+1} = f^n(x_1)$,
  where $f^n$ is the $n$th iterate of $f$,
  e.g., $f^3(x) = f(f(f(x)))$.
Thus, any $x_k$ in an $n$-cycle of $f$ must be a fixed point of $f^n$
\big[the reverse is, however, untrue, for a fixed point of $f^n$ 
  can also be a fixed point of $f^d$ as long as $d|n$: 
  if $f^d(x) = x$, then $f^n(x) = f^d(\cdots f^d(x)\cdots) = x$\big].
We can therefore classify a cycle as stable or unstable
  by the corresponding fixed point of $f^n$:
  a stable cycle requires
  $\big| \frac {d} {dx} f^n(x_1) \big| \le 1$,
  or by the chain rule,
\begin{equation}
  \Big| f'(x_1) \dots f'(x_n) \Big| \le 1,
\label{eq:der} 
\end{equation}
where $x_1. \ldots, x_n$ are the $n$ points within the cycle, or the orbit.
Further, the onset and bifurcation points
  are defined at the loci 
  where $\frac {d} {dx} f^n(x_1)$ reaches $+1$ and $-1$, 
  respectively \cite{strogatz}.

The outcome of the iterated sequence of course 
  depends on the parameter $r$.
Below we will present an algorithm to identify
  all regions of $r$ that allow stable $n$-cycles.
%

\section{\label{sec:logmap}Logistic map}

We will illustrate the algorithm on the logistic map \cite{may, strogatz},
  defined in \refeq{logmap}.
If $r$ is real,
  we will find windows $(r_a, r_b)$,
    within which stable $n$-cycles can exist.
%
Here, if $r > 0$, then
  $r_a$ ($r_b$) are the onset (bifurcation) points.
There are generally multiple such windows even for a single $n$,
but all onset points satisfy the same polynomial equation,
  and all bifurcation points another.
Our goal is thus to find the two polynomials for a given $n$.

To simplify the calculation, we first change variables \cite{saha} by
%
\[
    x^{(\mathrm{new})} \leftarrow r(x^{(\mathrm{old})} - 1/2),
    \quad R \leftarrow r(r-2)/4,
\]
%
and rewrite the map, in terms of $x^{(\mathrm{new})}$, as
\begin{equation}
  x_{k+1} = f(x_k) \equiv R - {x_k}^2.
\label{eq:logmaps}
\end{equation}
%
%
We will solve the cycle boundaries as zeros of the polynomials of $R$,
and the corresponding polynomials for $r$ can be obtained by $R \rightarrow r(r-2)/4$.
%
%

\subsection{Overall plan}

We solve the problem in two steps.
Since the $n$-cycles form a subset of the fixed points of $f^n$,
we will first find
  the polynomials at the stability boundaries of
  the fixed points of $f^n$
  (\refsecs{cyclic} to \ref{sec:examples}),
then remove contributions from shorter $d$-cycles ($d|n$)
  (\refsecs{primfac} and \ref{sec:degprimfac}).

Let us consider the first step of finding the fixed points of $f^n$.
At the first glance, the problem can be tackled by brute force:
  we can solve
  \refeqs{logmaps} and express $x_1, \ldots, x_n$
  in terms of $R$,
  and then plug the solution into \req{der}.
The result contains $R$ only (no $x_k$),
  and is therefore the answer.
But since \refeqs{logmaps} are nonlinear,
  it quickly becomes impossible for $n > 2$,
  as the degree of polynomials grows exponentially;
  thus it is nontrivial to reduce the final equation of $R$
  into a polynomial one.
Nonetheless, on a computer, one can 
  construct a Gr\"obner basis \cite{kk1}
  to automate the reduction.
The approach, albeit straightforward, does not 
  exploit the cyclic structure of \refeqs{logmaps},
can thus be improved by the following alternative.

Instead of solving \refeqs{logmaps} for $x_k$,
  we will derive a set of \emph{homogeneous linear} equations
  of cyclic polynomials of $x_k$
  (an example of a cyclic polynomial is $x_1 x_2 + x_2 x_3 + \dots + x_n x_1$).
%
  Now the matrix formed by the coefficients of the homogeneous linear equations 
  must have a zero determinant,
  for the cyclic polynomials are not zeros altogether.
%
Thus, the zero-determinant condition gives 
  the needed polynomial equations of $R$,
whose roots contain all fixed points of $f^n$.
This completes the first step.
%
%

For the second step, 
we show 
that short cycles serve as
factors in the polynomials obtained above,
and thus can be readily factored out.

\subsection{\label{sec:cyclic}Cyclic polynomials}

A polynomial is cyclic if it is invariant
  under the cycling of variables
  $x_1 \rightarrow x_2, x_2 \rightarrow x_3, 
  \ldots, x_n \rightarrow x_1$,
  e.g., $a(\vx) = x_1 x_2 + x_2 x_3 + x_3 x_4 + x_4 x_1$ 
  and $b(\vx) = x_1 x_3 + x_2 x_4$,
  for $n = 4$, where $\vx \equiv \{x_1, \ldots, x_n\}$.
A cyclic polynomial should not to be confused with a symmetric polynomial,
  which is invariant under the exchange of any two $x_k$ and $x_{j}$ ($j \ne k$);
  e.g., $a(\vx)$ and $b(\vx)$ are not symmetric,
  but $a(\vx) + b(\vx)$ is.

A cyclic polynomial can be generated by summing over
  distinct cyclic versions of a simpler polynomial of $x_k$, or a \emph{generator},
e.g., 
$x_1 x_2$ is a generator of $a(\vx)$;
$x_1 x_3$ is that of $b(\vx)$;
and
$x_1 x_2 + \frac{1}{2} x_1 x_3$ is that of $a(\vx) + b(\vx)$;
note that the coefficient before $x_1 x_3$ 
  is 1 in the second case
  for there are only two distinct versions,
  but is $\frac{1}{2}$ in the third case
  for there are four. 

We now consider cyclic polynomials
  generated from a monomial of unit coefficient,
  such as $a(\vx)$ and $b(\vx)$, but not $a(\vx) + b(\vx)$.
They can be systematically labeled as follows.
We pick the monomial generator,
  which can be written as 
  ${x_1}^{e_1} {x_2}^{e_2} \dots {x_n}^{e_n}$,
  then form a sequence $p$ of indices with
  $e_1$ 1's, $e_2$ 2's, \ldots, $e_n$ $n$'s;
  the corresponding cyclic polynomial is denoted by $C_p(\vx)$,
e.g.,
  $C_{12}(\vx)  = x_1 x_2 + x_2 x_3 + \dots + x_n x_1$ and
  $C_{112}(\vx) = {x_1}^2 x_2 + {x_2}^2 x_3 + \dots + {x_n}^2 x_1$.
We omit the length $n$ in this notation,
  for we will mostly work with a fixed $n$ at a time.
Since a cyclic polynomial can have multiple generators,
  e.g., both $x_1 x_2$ and $x_2 x_3$ are generators of $C_{12}(\vx)$
   (assuming $n \ge 3$),
  we pick the one that corresponds to
  the smallest $p$
  in the sense of lexicographic order,
  e.g., we choose $x_1 x_2$ instead of $x_2 x_3$ for $C_{12}(\vx)$, 
  because $12 < 23$.
Finally, we add $C_0(\vx) \equiv 1$ for completeness.

\subsection{Square-free cyclic polynomials}

We further restrict ourselves to 
  a subset of square-free cyclic polynomials,
  which have no square or higher powers of any $x_k$,
  e.g., $C_{12}(\vx) = x_1 x_2 + \dots$ is square-free,
  $C_{112}(\vx) = {x_1}^2 x_2 + \dots$ is not,
  see \reftab{sqrfreepoly} for more examples.
Obviously, the label $p$ of a square-free polynomial 
  has no repeated index.
We denote the set of all square-free $p$ by 
$\Bset = \{0, 1, 12, 13, \ldots, 123, 124, \ldots, 12\dots n\}$
such that its size $|\Bset|$ equals the number $\NB$ of
  square-free cyclic polynomials.

\begin{table}[h]\footnotesize
\caption{Square-free cyclic polynomials $C_p(\vx)$ for the logistic map ($n \ge 5$).}
\begin{center}
\begin{tabular}{c c c c c}
\hline
$p$ & $C_p(\vx)$ & Generator$^\dagger$ & Necklace$^\ddagger$ & $p$ as an index set$^*$\\ 
\hline 
0             & 1 & 1 & $0\ldots0$ & $\emptyset$ \\
1             & $x_1 + x_2 + \dots + x_n$ 
              & $x_1$ (or $x_2$, \ldots)
              & $10\ldots0$
              & $\{1\}$ (or $\{2\}$, \ldots) \\
12            & $x_1 x_2 + x_2 x_3 + \dots + x_n x_1$ 
              & $x_1 x_2$ (or $x_2 x_3$, \ldots)
              & $110\ldots0$
              & $\{1, 2\}$ (or $\{2, 3\}$, \ldots) \\
13            & $x_1 x_3 + x_2 x_4 + \dots + x_n x_2$
              & $x_1 x_3$ (or $x_2 x_4$, \ldots)
              & $1010\ldots0$
              & $\{1, 3\}$ (or $\{2, 4\}$, \ldots) \\
$\vdots$      & $\vdots$ \\
123           & \hspace{1mm}$x_1 x_2 x_3 + x_2 x_3 x_4 + \dots + x_n x_1 x_2$\hspace{1mm}
              & \hspace{1mm}$x_1 x_2 x_3$ (or $x_2 x_3 x_4$, \ldots)\hspace{1mm}
              & $1110\ldots0$
              & $\{1, 2, 3\}$ (or $\{2, 3, 4\}$, \ldots) \\
124           & $x_1 x_2 x_4 + x_2 x_3 x_5 + \dots + x_n x_1 x_3$
              & $x_1 x_2 x_4$ (or $x_2 x_3 x_5$, \ldots)
              & $11010\ldots0$
              & $\{1, 2, 4\}$ (or $\{2, 3, 4\}$, \ldots) \\
$\vdots$      & $\vdots$ \\
$12\dots n$  & $x_1 x_2 \dots x_n$ 
              & $x_1 x_2 \dots x_n$
              & $11\ldots1$
              & $\{1, 2, \ldots, n\}$  \\
\hline
\multicolumn{5}{p{\linewidth}}{
$^\dagger$ Alternative generators are shown in parentheses.
}\\
\multicolumn{5}{p{\linewidth}}{
$^\ddagger$ The corresponding binary necklaces.
}\\
\multicolumn{5}{p{\linewidth}}{
$^*$ The label $p$ of a square-free cyclic polynomial 
    has no repeated indices;
    so the indices can be cast to a set. 
}\\
\hline
\end{tabular}
\end{center}
\label{tab:sqrfreepoly}
\end{table}

We first show that 
  the square-free cyclic polynomials 
  serve as a basis for expanding cyclic polynomials:

\begin{theorem}
  For the logistic map \refeq{logmaps}, 
  any cyclic polynomial $K(\vx)$
  formed by the $n$-cycle points
  $\vct x = \{x_1, \ldots, x_n\}$
  is a linear combination of 
the square-free cyclic polynomials $C_p(\vx)$:
\[
  K(\vx) = \sum_{p \in \Bset} f_p(R) \, C_p(\vx),
\]
  where $\Bset = \{0, 1, 12, 13, \ldots, 12\dots n\}$ is
  the set of indices of all square-free cyclic polynomials,
  and $f_p(R)$ are polynomials of $R$.
  \label{thm:sqrfree}
\end{theorem}

\begin{proof}
We show the theorem by the following square-free reduction.
Given a cyclic polynomial $K(\vx)$,
  we recursively apply \refeq{logmaps} as
  $x_k^2 \rightarrow R - x_{k+1}$,
  until all squares or higher powers of $x_k$ are eliminated.
The process will not last indefinitely
  for each substitution reduces the degree in $x_k$ (no matter which $k$)
  by one.
Since the original polynomial is cyclic,
  so is the reduced one.
All terms that involve no $x_k$
  are collected to serve as the coefficient before $C_0(\vx)$, which is 1.
Since no square or higher powers of $x_k$
  can survive the reduction,
  all cyclic polynomials in the final result 
  are square-free.
The coefficients are polynomials of $R$,
  for $R$ is the only variable introduced by the substitutions.
\end{proof}

For example, for $n = 2$, the cyclic polynomial
  $K(\vx) = {x_1}^2 \, x_2 + {x_2}^2 \, x_1$
can be written as $K(\vx) = (R + 1) \, C_{1}(\vx) - 2 R \, C_0(\vx)$
for
${x_1}^2 \, x_2 = R \, x_2 - {x_2}^2 = R \, x_2 - R + x_1$ and
${x_2}^2 \, x_1 = R \, x_1 - R + x_2$.

\refthm{sqrfree} 
  shows that any cyclic polynomial
  can be expanded
  as a combination of the square-free ones,
  which serve as a basis.
Below we show that at the onset and bifurcation points,
  the square-free cyclic polynomials $C_p(\vx)$ are themselves
  linearly connected by an $\NB \times \NB$ matrix equation.
The determinant of matrix must vanish,
  and this condition yields the solution of the problem.

\subsection{\label{sec:algo}Algorithm for locating fixed points of $f^n$}

We first observe that the derivative of $f^n$
is a cyclic polynomial:
\begin{equation}
  \Lambda(\vx)
   = \frac{d}{dx}f^{n}(x_1)
   = f'(x_n) \dots f'(x_1)
   = (-2)^n x_1 \dots x_n.
\label{eq:logder}
\end{equation}

Now, for any $p$, $\Lambda(\vx) \, C_p(\vx)$
  is also a cyclic polynomial,
since the product of two cyclic polynomials is cyclic too.
We can therefore expand it by \refthm{sqrfree} as
\begin{equation}
  \Lambda(\vx) \, C_p(\vx) = \sum_{q \in \Bset} T_{pq}(R) \, C_q(\vx),
\label{eq:xcp0}
\end{equation}
where $T_{pq}(R)$ is a polynomial of $R$,
and $p, q \in \Bset$.

By \refeq{der}, at the onset or bifurcation point,
  $\Lambda(\vx)$ is equal to a number
  $\lambda = +1$ or $-1$, respectively;
  so \refeq{xcp0}
  becomes a homogeneous linear equation of
  $C_p(\vx)$:
\begin{equation}
  \lambda \, C_p(\vx) = \sum_{q \in \Bset} T_{pq}(R) \, C_q(\vx),
\label{eq:xcp}
\end{equation}
%
%
or in matrix form,
\begin{equation}
  \big[ \lambda \, \vct I - \vct T(R) \big] \, \vct C = 0,
\tag{$\ref{eq:xcp}'$}
\end{equation}
where
$\vct I$ is the $\NB \times \NB$ identity matrix,
$\vct T(R) = \{T_{pq}(R)\}$ is an $\NB \times \NB$ matrix,
and
$\vct C = \{C_p(\vx)\}$ is an $\NB$-dimensional column vector.

Since a set of homogeneous linear equations
  has a non-trivial solution 
  only if the determinant of the coefficient matrix
  is zero, we have
\begin{equation} 
  A_n\RX \equiv \Big| \lambda \, \vct I - \vct T(R) \Big| = 0.
\label{eq:xr}
\end{equation}
Here we have defined $A_n\RX$ as a polynomial of $R$ and $\lambda$,
  and we have also attached the subscript $n$, 
  for later use with $A_d\RX$, where $d$ are divisors of $n$.
\refeq{xr} is a necessary condition 
  since $C_p(\vx)$ cannot vanish altogether;
and since it involves $R$ only, 
  the polynomial expansion of the determinant
  gives the answer to our problem.

To summarize, we have
\begin{theorem}
  At the onset and bifurcation points,
  the square-free cyclic polynomials $C_p(\vx)$
  are linear related by \refeq{xcp}, with
  $\lambda$ being $+1$ and $-1$, respectively,
  and $T_{pq}(R)$ the coefficients
  from the square-free reduction of $\Lambda(\vx) \, C_p(\vx)$
  with $\Lambda(\vx)$ specified by \refeq{logder}.
  Thus, $R$ at the two points are the roots of the 
  polynomials $A_n(R, \lambda = \pm1)$
  obtained from the characteristic equation \refeq{xr}.
  \label{thm:main}
\end{theorem}

\begin{remark}[1]
  For complex $R$ and $\vx$,
  $\lambda$ should be generalized from $\pm1$
  to any $\lambda = \exp(i\phi)$,
  where $\phi \in [0, 2\pi]$,
  the algorithm still applies.
  By increasing $\phi$ from 0 to $2\pi$,
  we can trace a two-dimensional region of a complex $R$ 
  for stable cycles.
  These regions are bulbs in the Mandelbrot set,
  see ref. \cite{stephenson} and \reffig{logbulb}.
  
  Further, with $\lambda = 0$,
  the algorithm determines the superstable point,
  at which the deviation from a cycle point vanishes
  to the linear order after $n$ iterations of $f$.
  But we have a better algorithm in this case:
  since at least one of the $x_k$ is zero by \refeq{logder},
  then $f^n(x_k = 0) = x_{n+k} = 0$ provides 
  the needed polynomial equation of $R$ \cite{strogatz}.
\end{remark}

\subsection{\label{sec:examples}Examples}

We illustrate the above algorithm by cases of small $n$.
It is still helpful to have a mathematical software verify some steps
  (e.g., in computing the determinants and their factorization).

For $n = 1$, we have two square-free cyclic polynomials
  $C_0(\vx) = 1$ and $C_1(\vx) = x_1$;
  and $\Lambda(\vx) = -2 \, x_1$
\big[we shall drop ``$(\vx)$'' below for convenience\big].
Thus $\Lambda \, C_0 = -2 \, x_1 = -2 \, C_1$,
$\Lambda \, C_1 = -2 \, x_1^2 = -2 \, R + 2 x_1 = -2 \, R \, C_0 + 2 \, C_1$,
or
\[
  \lambda
  \left( \begin{array}{c}
  C_0 \\
  C_1
  \end{array} \right)
  =
  \left( \begin{array}{cc}
  0     & -2 \\
  -2R   &  2
  \end{array}\right)
  \left( \begin{array}{c}
  C_0 \\
  C_1
  \end{array} \right),
\]
and \refeq{xr} reads
\begin{equation}
0=  \left| \begin{array}{cc}
  \lambda     & 2          \\
  2 R         & \lambda - 2
  \end{array} \right| = -4R - 2\lambda + \lambda^2,
  \tag{\ref{eq:xr}-1}
\label{eq:xr1}
\end{equation}
which is just the equation for a fixed point.
The fixed point begins at $\lambda = +1$ or $R_a = -1/4$,
  and becomes impossible when $\lambda = -1$, or $R_b = 3/4$.

For $n = 2$, the cyclic variables are
$C_0 = 1$,
$C_1 = x_1 + x_2$,
$C_{12} = x_1 x_2$,
and $\Lambda = 4 x_1 x_2$.
Thus,
$\Lambda \, C_0 = 4 C_{12}$,
$\Lambda \, C_1 = 4 \, {x_1}^2 \, x_2 + 4 \, x_1 \, {x_2}^2
  = 4 \, R \, (x_1 + x_2) - 4 ({x_1}^2 + {x_2}^2)
  = 4 \, R \, (x_1 + x_2) - 8 \, R + 4 \, (x_1 + x_2)
  = 4 \, (R + 1) \, C_1 - 8 \, R \, C_0$,
$\Lambda \, C_{12}
  = 4 \, (R-x_2)\, (R-x_1) 
  = 4\, R^2  \, C_0 - 4 \, R \, C_1 + 4 \, C_{12}$,
or
\[
\lambda
  \left( \begin{array}{c}
  C_0 \\
  C_1 \\
  C_{12}
  \end{array} \right)
 =
  \left( \begin{array}{ccc}
  0           & 0         & 4 \\
  -8R         & 4(R+1)    & 0 \\
  4R^2        & -4R       & 4
  \end{array} \right)
  \left( \begin{array}{c}
  C_0 \\
  C_1 \\
  C_{12}
  \end{array} \right),
\]
and \refeq{xr} reads
\begin{align}
0 &=
  \left| \begin{array}{ccc}
  \lambda     & 0               & -4 \\
  8R          & \lambda-4(R+1)  & 0 \\
  -4R^2       & 4R              & \lambda-4
  \end{array} \right| \nonumber\\
  &= (4R - 4 + \lambda) [(4 R-\lambda)^2 - 4 \lambda].
  \tag{\ref{eq:xr}-2}
\label{eq:xr2}
\end{align}
We only use the first factor
(the choice will be explained later, same for the following cases).
Setting it to zero yields $R = 1 - \lambda/4$;
$\lambda = +1$ gives the onset value $R_a = 3/4$ ($r_a = 3$)
while
$\lambda = -1$ gives the bifurcation value $R_b = 5/4$ ($r_b = 1+\sqrt 6$).
Note that the onset point of the only 2-cycle 
  is located at $R = 3/4$, 
  where the fixed point bifurcates \cite{strogatz}
  \big[compare \reffigs{cobweb}(a) and (b)\big].

For $n = 3$ \cite{saha, bechhoefer, gordon, burm, zhang},
we have
$C_0 = 1$,
$C_1 = x_1 + x_2 + x_3$,
$C_{12} = x_1 x_2 + x_2 x_3 + x_3 x_1$,
$C_{123} = x_1 x_2 x_3$,
and
$\Lambda = -8 x_1 x_2 x_3$.
The square-free reduction yields
\begin{align}
\lambda
\left( \begin{array}{c}
  C_0 \\
  C_1 \\
  C_{12} \\
  C_{123}
\end{array} \right)
=
\left( \begin{array}{cccc}
 0            & 0                 & 0               & -8 \\
 -24 R        & 8(R+1)            & -8R             & 0 \\
 24 R^2       & -8R(R+2)          & 8(R+1)          & 0 \\
 -8R^3        & 8R^2              & -8R             & 8
  \end{array} \right) 
\left( \begin{array}{c}
  C_0 \\
  C_1 \\
  C_{12} \\
  C_{123}
\end{array} \right),
\end{align}
and \refeq{xr} reads:
\begin{align}
  0
&=\left| \begin{array}{cccc}
 \lambda      & 0                 & 0               & 8 \\
 24 R         & \lambda - 8(R+1)  & 8R              & 0 \\
 -24 R^2      & 8R(R+2)           & \lambda-8(R+1)  & 0 \\
 8R^3         & -8R^2             & 8R              & \lambda-8
  \end{array} \right| \nonumber \\
&= -\big[ 
  64R^3 - 128 R^2 - 8 (\lambda - 8) R - (\lambda -8)^2 
  \big]
  \big(
  \lambda^2 - 8\lambda - 24 R\lambda - 64 R^3
  \big).
  \tag{\ref{eq:xr}-3}
\label{eq:xr3}
\end{align}
Using the first factor, we find at the onset point
  $\lambda = 1$, 
  $\left(
  R-\frac74
  \right)
  \left(
  R^2-\frac14 R + \frac{7}{16}
  \right)=0$
  and its only real solution is $R_a = 7/4$ ($r_a = 1+\sqrt 8$).
At the bifurcation point $\lambda = -1$,
  the equation $R^3-2R^2 + \frac{9}{8} R -\frac{81}{64}=0$
  yields
  $R_b = \frac14 
      \left(
      \frac{8}{3} +
        \sqrt[3]{\frac{1915}{54} - \frac52\sqrt{201}}
       +\sqrt[3]{\frac{1915}{54} + \frac52\sqrt{201}}
      \right)$,
whose corresponding $r = 1+\sqrt{1+4R}$
is identical to that in ref. \cite{gordon, burm}.

For $n = 4$ \cite{stephenson}, the cyclic variables are
$C_0 = 1$,
$C_1 = x_1 + x_2 + x_3 + x_4$,
$C_{12} = x_1 x_2 + x_2 x_3 + x_3 x_4 + x_4 x_1$,
$C_{13} = x_1 x_3 + x_2 x_4$,
$C_{123} = x_1 x_2 x_3 + x_2 x_3 x_4 + x_3 x_4 x_1 + x_4 x_1 x_2$
$C_{1234} = x_1 x_2 x_3 x_4$,
and
$\Lambda = 16 x_1 x_2 x_3 x_4$.
\refeq{xr} reads
\begin{align}
0 & = \left|
\begin{array}{cccccc}
 \lambda  & 0               & 0                 & 0     & 0     & -16 \\
 64 R     & \lambda-16(R+1) & 16 R              & 0     & -16R  & 0 \\
 -64 R^2  & 16R(R+2)        & \lambda-16(R^2+1) & -32 R & 16R   & 0 \\
 -32 R^2  & 16R(R+1)        & -16R              & \lambda-16(R^2+1)  & 0 & 0 \\
 64 R^3   & -16R^2(R+3)     & 16R(R+2)        & 32 R (R+1)       & \lambda-16(R+1) & 0 \\
 -16R^4   & 16R^3           & - 16R^2           & -16R^2            & 16R & \lambda - 16
\end{array}
\right | \nonumber  \\
& =
    \big[
    4096 R^6-12288 R^5 + 256 (\lambda + 48) (R^4- R^3)
      -16 (\lambda + 32)(\lambda - 16) R^2
      -(\lambda - 16)^3 \big]  \nonumber \\
& \quad
    \big[ 16 (R-1)^2 - \lambda \big]
 \, \big[(16 R^2 + \lambda)^2 - 16 (2R+1)^2 \lambda \big],
  \tag{\ref{eq:xr}-4}
\label{eq:xr4}
\end{align}
From the first factor, we have
$ ( 
    4R - 5
  )
\, \big[
  (4R + 1)^2 + 4
  \big]
\, \big[
  (4R - 3)^3 - 108
  \big] =0$
at the onset point $\lambda = 1$.
It has two real roots:
$R_a' = 5/4$ 
\big($r_a' = 1+\sqrt{6} \approx 3.4495$\big)
for the cycle from period-doubling the 2-cycle 
  \big[compare \reffigs{cobweb}(b) and (d)\big],
and
$R_a = (3+\sqrt[3]{108})/4$
\big($r_a = 1+\sqrt{4+\sqrt[3]{108}} \approx 3.9601$\big)
for an original cycle \big[\reffig{cobweb}(e)\big].
At the bifurcation point, $\lambda = -1$, and
$4096 R^6 - 12288 R^5 + 12032 (R^4 - R^3)
  + 8432 R^2 + 4913 = 0$,
which upon $R \rightarrow r(r-2)/4$ yields the same polynomial
obtained previously
\cite{bailey1, kk1, bailey2, lewis}.
The only two positive roots $r_b \approx 3.9608$
and $r_b' \approx 3.5441$
correspond to $r_a$ and $r_a'$ respectively.
As a verification, the polynomials are
alternatively derived in \refapd{per4}.

\begin{table}[h]\footnotesize
  \caption{
  Characteristic polynomials $A_n\RX$ of 
  the fixed points of $f^n$ of
  the simplified logistic map \refeq{logmaps}.
  }
\begin{center}
\begin{tabular}{lc}
\hline
  $n$ 
& $A_n(R, \; \lambda = 2^n X) 
  \big\slash
   2^{n \NB(n)}$ $^\dagger$  
\\
\hline
1
&
$X^2-X-R$
\\
2
&
  $
  \big[ \cancel{(R-X)^2 - X}\big]
  (R + X - 1)
  $
\\
3
&
$\begin{aligned}
\big[\cancel{
  X^2-(3R+1)X-R^3
}\big]
  \big[ 
  (X-1)^2
 + X R 
 - R(R-1)^2
  \big] 
\end{aligned}$
\\
4
&
$\begin{aligned}
\big[ \cancel{
  (R-1)^2 - X
} \big]
\big[ \cancel{ 
(R^2+X)^2 -(2R+1)^2 X 
} \big]
\big[
  R^6-3 R^5+(X+3) (R^4 - R^3) -(X+2)(X-1) R^2 - (X-1)^3
\big]
\end{aligned}$
\\
5
&
\begin{minipage}{.97\linewidth}
\vspace*{2mm}
$\begin{aligned}
&\big[ \cancel{
  X^2-5 X (R^2+R)-X - R^5
} \big]
\big[
-R^{15}+8 R^{14}-28 R^{13}+(X+60) R^{12}-(7 X+94) R^{11} 
+(3 X^2 +20 X
\\
+&116) R^{10}
-(11 X^2+33 X+114) R^9 
+2(3 X^2+20 X+47) R^8
-(2 X^3-20 X^2+37 X+69) R^7 \\
+&(3 X-11)(3X^2-3X+44) R^6
-(X-1) (3 X^3+20 X^2-33 X-26) R^5 
+(X-1)^2 (3 X^2+27 X+14) R^4 \\
+&(X-1)^3 (6 X+5) R^3 
+(X-1)^4 (X+2) R^2+(X-1)^5 R+(X-1)^6
\big]
\end{aligned}$
\vspace{1mm}
\end{minipage}
\\
6
&
$\begin{aligned}
&\big[\cancel {
  (R-1)^3 + X
}\big]
\big[\cancel {
 (R^3 - X)^2 - (3R+1)^2 X 
}\big]
\big[\cancel {
\big(R(R-1)^2-1-X\big)^2-(R-2)^2 X
}\big]
\big[
R^{27}-13 R^{26}+78 R^{25}\\
&+(X-293) R^{24} 
+\dots
+(X-1)^6 (X^2+10 X+3) R^3
+(X-1)^7 (X+1) R^2-(X-1)^8 R+(X-1)^9
\big]
\end{aligned}$
\\
7
&
\begin{minipage}{.97\linewidth}
\vspace{1mm}
$\begin{aligned}
&\big[\cancel {
  -R^7 - 7 X R (R+1)^2 + X^2 - X
}\big]
\big[
-R^{63}+32 R^{62}-496 R^{61}
+4976 R^{60} + (X - 36440) R^{59}
- (30X
\\
-&208336) R^{58}
+\dots
+(X-1)^{15} (X^2+14 X+5) R^3
+2(X-1)^{16} (X+1) R^2
+(X-1)^{17} R
+(X-1)^{18}
\big]
\end{aligned}$
\end{minipage}
\\
\vdots & \vdots
\\
\hline
\multicolumn{2}{p{\textwidth}}{
$^\dagger$
$\NB(n) = (1/n) \sum_{d|n} \phi(n/d) 2^d$ [\refeq{necklace}]
is the number of the square-free cyclic polynomials.
The change of variable $\lambda \rightarrow X$ and 
  the division by $2^{n \NB(n)}$ make the polynomials more compact. 
The irrelevant factors from shorter cycles 
(see \refsec{primfac}) are struck out.
The polynomials of $r$ for the original logistic map \refeq{logmap}
can be obtained by $R\rightarrow r(r-2)/4$.
} \\
\hline
\end{tabular}
\end{center}
\label{tab:Anlog}
\end{table}

The algorithm was coded into a Mathematica program, 
which was used to compute the polynomials for $n$ up to 13.
The polynomials for a general $\lambda$
and those at $\lambda = \pm1$ (onset and bifurcation points) are listed
in \reftab{Anlog} and \reftab{Pnlog}, respectively, for some small $n$.
For complex $R$, $\lambda$, and $\vx$,
the method can also compute the region of stability for $R$,
with $\lambda$ being $\exp(i\phi)$ \big[$\phi \in (0, 2\pi)$\big] 
  instead of $\pm1$;
the results are shown in \reffig{logbulb} 
for $n$ up to $8$.
For polynomials of larger $n$, see the website in \refsec{end}.
The representative $r$ values are listed in \reftab{rval}.

\begin{table}[h]\footnotesize
  \caption{
  Smallest positive $r$ at the onset and bifurcation points
  of the $n$-cycles of the logistic map.
  }
\begin{center}
\begin{tabularx}{\textwidth}{
  >{\hsize=0.5\hsize\centering\arraybackslash}X  
  >{\hsize=1.6\hsize}X  
  >{\hsize=1.6\hsize}X  
  >{\hsize=0.3\hsize\raggedright\arraybackslash}X | 
  >{\hsize=0.5\hsize\centering\arraybackslash}X  
  >{\hsize=1.6\hsize}X  
  >{\hsize=1.6\hsize}X  
  >{\hsize=0.3\hsize\raggedright\arraybackslash}X  
}
\hline
  $n^\dagger$ 
& Onset$^\ddagger$  
& Bifurcation$^\ddagger$
& \#$^*$ 
&
  $n^\dagger$     
& Onset$^\ddagger$
& Bifurcation$^\ddagger$
& \#$^*$ \\
\hline
$1$     & $1.0000000000_1$      &  $3.0000000000_1$       & 1   &
$8'$    & $3.9607686524_{6}$    &  $3.9610986335_{120}$   & 1   \\
$2'$    & $3.0000000000_1$      &  $3.4494897428_1$       & 1   &
$8'''$  & $3.5440903596_{6}$    &  $3.5644072661_{120}$   & 1   \\
$3$     & $3.8284271247_1$      &  $3.8414990075_3$       & 1   &
$9$     & $3.6871968733_{240}$  &  $3.6872742105_{252}$   & 28  \\
$4$     & $3.9601018827_3$      &  $3.9607686524_6$       & 1   &
$10$    & $3.6052080669_{472}$  &  $3.6059169323_{495}$   & 48  \\
$4''$   & $3.4494897428_1$      &  $3.5440903596_6$       & 1   &
$10'$   & $3.7411207566_{15}$   &  $3.7425706462_{495}$   & 3   \\
$5$     & $3.7381723753_{11}$   &  $3.7411207566_{15}$    & 3   &
$11$    & $3.6817160194_{1013}$ &  $3.6817266457_{1023}$  & 93  \\
$6$     & $3.6265531617_{20}$   &  $3.6303887000_{27}$    & 4   &
$12$    & $3.5820230011_{1959}$ &  $3.5828117795_{2010}$  & 165 \\
$6'$    & $3.8414990075_{3}$    &  $3.8476106612_{27}$    & 1   &
$12'$   & $3.6303887000_{27}$   &  $3.6321857392_{2010}$  & 4   \\
$7$     & $3.7016407642_{57}$   &  $3.7021549282_{63}$    & 9   &
$12''$  & $3.8476106612_{27}$   &  $3.8490363152_{2010}$  & 1   \\
$8$     & $3.6621089132_{108}$  &  $3.6624407072_{120}$   & 14  &
$13$    & $3.6797024578_{4083}$ &  $3.6797038498_{4095}$  & 315 \\
\hline
\multicolumn{8}{p{\textwidth}}{
$^\dagger$
  $\,'$, $\,''$, or $\,'''$ means
    a cycle undergoing
    the first, second, or third successive period-doubling, respectively.
} \\
\multicolumn{8}{p{\textwidth}}{
$^\ddagger$
  The subscripts are the degrees of the corresponding minimal polynomial
    of $R = r(r-2)/4$.
} \\
\multicolumn{8}{p{\textwidth}}{
$^*$
  The number of similar cycles.
} \\
\hline
\end{tabularx}
\end{center}
\label{tab:rval}
\end{table}

\subsection{\label{sec:primfac}Minimal polynomial for the $n$-cycles}

The factors ignored in \refsec{examples}
  come from shorter $d$-cycles whose periods $d$ divide $n$,
  because $A_n\RX$, 
  from the characteristic equation \refeq{xr},
  is derived for all fixed points of $f^n$,
  and thus encompasses the shorter cycles as well.
We filter the contributions from the shorter cycles by the following theorem.

\begin{theorem}
  The minimal polynomial $P_n\RX$ of all $n$-cycles
  is a factor of $A_n\RX$ [defined in \refeq{xr}],
  and can be computed as
  \begin{equation}
    P_n\RX
    = \prod_{cd = n} B_{d,c}\RX^{\mu(c)},
  \label{eq:primfac}
  \end{equation}
where
  $B_{d, c}\RX \equiv \prod_{k=1}^c A_{d}(R, e^{2k\pi i/c} \lambda^{1/c})$,
  $\lambda^{1/c}$ is a complex $c$th root of $\lambda$,
  and $\mu(c)$ is the M\"obius function.
  \label{thm:primfac}
\end{theorem}

\begin{remark}[1]
The M\"obius function $\mu(n)$ is $(-1)^k$
  if $n$ is the product of $k$ distinct primes,
  or 0 if $n$ is divisible by a square of a prime.
$\mu(n)$ = 1, $-1$, $-1$, 0, $-1$, 1, \ldots, starting from $n = 1$.
The $\mu(n)$ is useful for inversion:
$g(n) = \sum_{d|n} \mu(n/d) \, h(d)$
if and only if $h(n) = \sum_{d|n} g(d)$
\cite{hardy}.
\end{remark}

\begin{remark}[2]
$B_{d,c}\RX$ is a polynomial of $\lambda$.
Despite the argument $\lambda^{1/c}$,
  the product $\prod_{k=1}^c A_d(R, e^{2k\pi i/c} \lambda^{1/c})$
  is free from radicals of powers of $\lambda^{1/c}$,
  for it is invariant under $\lambda \rightarrow e^{2\pi i} \lambda$;
and $\deg_\lambda B_{d,c}\RX = \deg_\lambda A_d\RX$.
Particularly, $B_{n,1}\RX = A_n\RX$.
\end{remark}


Let us see some examples.
For $n = 1$,
there is no irrelevant factor
in \refeq{xr1} and
$P_1\RX = B_{1,1}\RX = A_1\RX = \lambda^2 - 2\lambda - 4R$.

For $n = 2$, since
$B_{1,2}\RX
=
(\lambda +2\sqrt{\lambda} - 4R)
\,
(\lambda -2\sqrt{\lambda} - 4R)
=(4R-\lambda)^2 -4\lambda$,
$P_2\RX = A_2\RX \, B_{1,2}\RX^{-1} = 4R - 4 + \lambda$.

For $n = 3$,
one can verify that 
$B_{1,3}\RX 
= \prod_{k=1}^3 A_1(R, e^{2k\pi i/3} \, \sqrt[3]{\lambda})
=  \lambda^2 - 24 R \lambda - 8 \lambda - 64 R^3$.
So
$P_3\RX = A_3\RX \, B_{1,3}\RX^{-1} 
= -\big[64 R^3- 128 R^2 -8(\lambda -8)R-(\lambda-8)^2\big]$.

For $n = 4$,
we have
$16 (R-1)^2 - \lambda
= (4R - 4 + \sqrt{\lambda})(4R - 4 - \sqrt{\lambda})$
and
$
(16 R^2 + \lambda)^2 - 16 (2R+1)^2 \lambda
=
\big[(4R - \sqrt{\lambda})^2  - 4\sqrt{\lambda}\,\big]
\big[(4R + \sqrt{\lambda})^2  + 4\sqrt{\lambda}\,\big]$.
Thus, the last two factors of \refeq{xr4}
can be written as
$B_{2,2}\RX = \prod_{k=1}^2 A_2(R, e^{k \pi i} \sqrt{\lambda})$,
and
$P_4\RX  
    = A_4\RX \, B_{2,2}\RX^{-1}
    = 4096 R^6-12288 R^5+ 256 (\lambda + 48) (R^4- R^3)
    -16(\lambda + 32)(\lambda - 16) R^2 - (\lambda-16)^3$.
Note, $B_{1,4}\RX$ is unused for $\mu(4) = 0$.

The irrelevant factors for $n$ up to 7 are listed in \reftab{Anlog}.

\newcommand{\T}{{R_4}}

\begin{table}[h]\footnotesize
  \caption{
  Onset and bifurcation polynomials 
  of the $n$-cycles of the simplified logistic map \refeq{logmaps}.}
\begin{center}
\begin{tabular}{lcc}
\hline
  $n$ 
& Onset $P_n(R, \; +1)$ $^{\dagger, \ddagger}$
& Bifurcation $P_n(R, \; -1)$ $^\dagger$
\\
\hline
1
&
$-\T-1$
&
$-\T+3$
\\
2
&
$-(\cancel{3-\T})$
&
$\T-5$
\\
3
&
$\begin{aligned}
-
\big( \cancel{
  \T^2-\T+7
}\big)
  (\T-7) 
\end{aligned}$
&
$\begin{aligned}
 - \T^3 + 8\T^2 - 18\T + 81
\end{aligned}$
\\
4
&
$\begin{aligned}
\big( \cancel{
  \T - 5
} \big)
\big[ \cancel{ 
  (\T+1)^2 + 4
} \big]
\big[
  (\T-3)^3 - 108
\big]
\end{aligned}$

&

$
\T^6-12 \T^5+47 \T^4-188 \T^3+527 \T^2+4913
$
\\
5

&

\begin{minipage}{.48\linewidth}
$\begin{aligned}
-
&\big( \cancel{
\T^4-\T^3+\T^2+9 \T+31
} \big)
\big(
\T^{11}-31 \T^{10} \\
+&416 \T^9 -3404 \T^8 +20548 \T^7-98258 \T^6 \\
+&370146 \T^5 -1171676 \T^4 +3301996 \T^3\\
-&7507332 \T^2+15699857 \T-28629151
\big)
\end{aligned}$
\end{minipage}

&

\begin{minipage}{.49\linewidth}
\vspace{1mm}
$\begin{aligned}
-&\T^{15}+32 \T^{14}-448 \T^{13}+3838 \T^{12}-24008 \T^{11}\\
+&118147 \T^{10} -462764 \T^9+1519712 \T^8\\
-&4444924 \T^7 +11351480 \T^6 -26978787 \T^5\\
+&58697100 \T^4 -88548768 \T^3+149426046 \T^2\\
-&313083144 \T +1291467969
\end{aligned}$
\vspace{1mm}
\end{minipage}

\\
6

&

\begin{minipage}{.48\linewidth}
$\begin{aligned}
&\big( \cancel{
\T^2-9 \T+21
} \big)
\big( \cancel{
\T^2+3 \T+3
} \big)
\big( \cancel{
\T^3-8\T^2} \\ 
+&\cancel{18\T-81
} \big)
\big(
\T^{20}-\dots+3063651608241
\big)
\end{aligned}$
\end{minipage}

&

\begin{minipage}{.49\linewidth}
$\begin{aligned}
&\T^{27}-52 \T^{26}+1248 \T^{25}-18753 \T^{24}+\dots\\
&-5098317006250000 \T-20711912837890625
\end{aligned}$
\end{minipage}

\\
7

&

\begin{minipage}{.48\linewidth}
$\begin{aligned}
-
&\big( \cancel{
\T^6-\T^5+\T^4-\T^3
 +15T^2+97 \T+127
} \big) \\
&\big(
\T^{57}
-127\T^{56}
+\dots+
58165204\dots8504447
\big)
\end{aligned}$
\end{minipage}

&

\begin{minipage}{.49\linewidth}
$\begin{aligned}
&-\T^{63}+128 \T^{62}-7936 \T^{61}+318464 \T^{60}-\dots\\
&-2427583\dots6441888 \T +9786215\dots0031361
\end{aligned}$
\end{minipage}

\\
\vdots & \vdots & \vdots
\\
13
&

\begin{minipage}{.48\linewidth}
$\begin{aligned}
-&\big(\cancel{
\T^{12}-\T^{11}+\T^{10}
-\dots
+18433 \T+8191
}\big) \\
\big(
&\T^{4083}
 - 8191\T^{4082} 
+ 33529856 \T^{4081}
- \dots \\
-&30826655683291995
\dots
27828275886475354111
\big)
\end{aligned}$
\end{minipage}

&

\begin{minipage}{.49\linewidth}
$\begin{aligned}
-&\T^{4095} + 8192 \T^{4094} -33538048 \T^{4093} \\
+&\dots 
-7361199006999 
\dots
96207964555264\T
\\
+&29448390363448812
\dots
352154556569141249
\end{aligned}$
\end{minipage}

\\
\hline
\multicolumn{3}{p{\textwidth}}{
$^\dagger$ $\T = 4R$. 
The polynomials of $r$ for the original logistic map \refeq{logmap}
can be obtained by $\T\rightarrow r(r-2)$.
} \\
\multicolumn{3}{p{\textwidth}}{
$^\ddagger$
Factors for the $n$-cycles born out of shorter cycles (see \refsec{origfac})
are struck out.
} \\
\hline
\end{tabular}
\end{center}
\label{tab:Pnlog}
\end{table}

\begin{figure}[h]
  \begin{minipage}{\linewidth}
    \begin{center}
        \includegraphics[angle=-90, width=.67\linewidth]{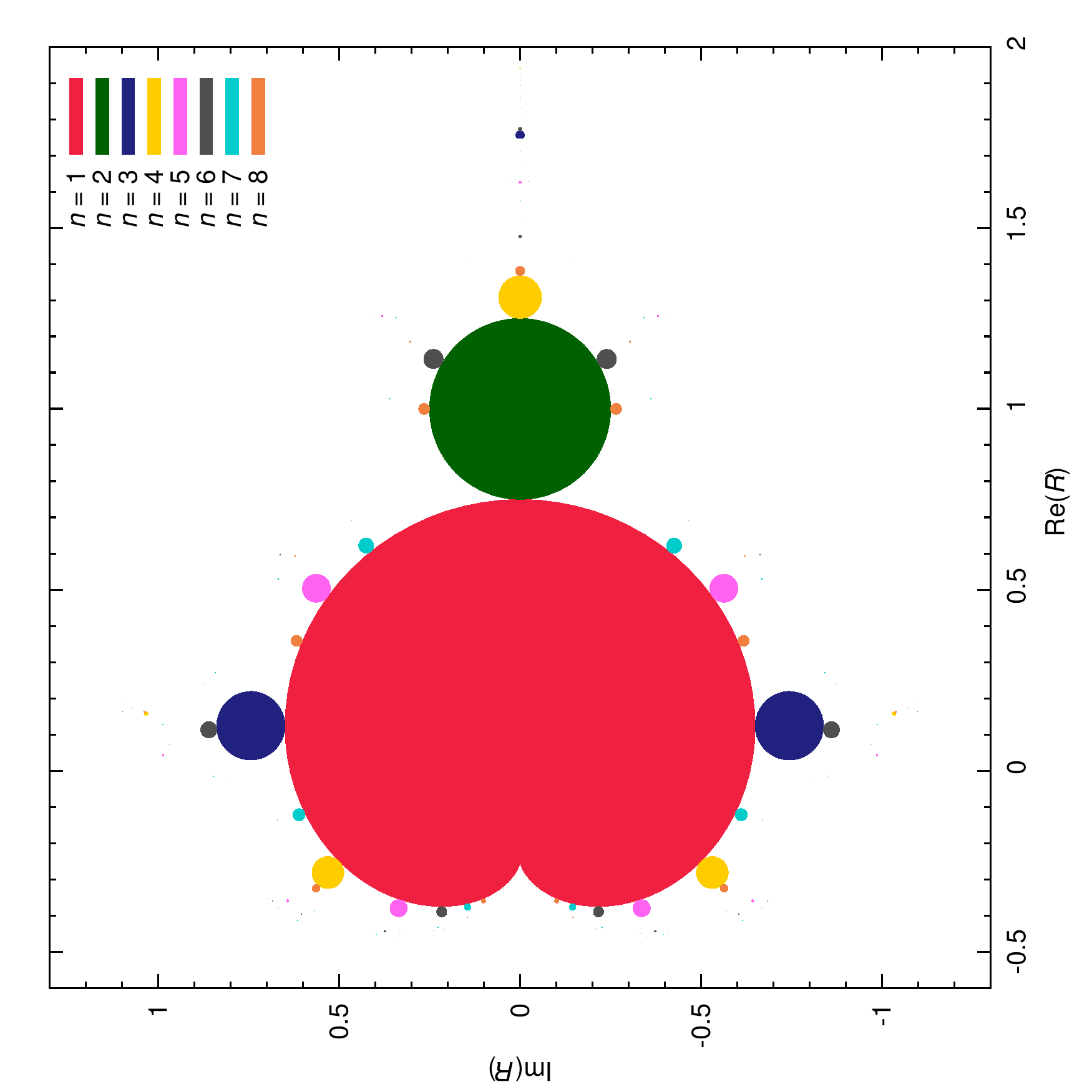}
    \end{center}
  \end{minipage}%
  \caption{
  \label{fig:logbulb}
  Stable regions of the $n$-cycles of the simplified logistic map \refeq{logmaps} 
  with a complex $R$;
  obtained from $P_n(R, e^{i\phi}) = 0$ with $\phi \in [0, 2\pi]$
  cf. Fig. 3 in \cite{stephenson}.
}
\end{figure}

\refthm{primfac} is not always necessary.
For $n \ge 4$, 
  $P_n\RX$ is readily recognized as
  the factor
  of $A_n\RX$ 
  with the highest degree in $R$,
  see \reftab{Anlog} and \refsec{degR}.
It can be, however, problematic,
   if $P_n(R, \lambda)$ is solved for $\lambda = 1$ 
   instead of a general $\lambda$,
   for $P_n(R, 1)$ itself can be further factorized,
  see \reftab{Pnlog}.
Due to the technical nature of the derivation and subsequent discussions, 
  the reader may wish to skip the rest of \refsec{logmap}
  on first reading.

\subsection{\label{sec:degprimfac}Counting cycles}

To show \refthm{primfac}, we first find the degrees in $\lambda$ of 
  $A_n\RX$ (\refthm{necklace}) and $P_n\RX$ (\refthm{lyndon}).
By comparing the degrees,
  we then show that each $P_d\RX$ with ($d|n$),
  after some transformation,
  contributes one polynomial factor to $A_n\RX$
  (\refthm{prod}),
  and the inversion of the relation
  yields \refthm{primfac}.

\subsubsection{Number of the square-free cyclic polynomials}

To count the square-free cyclic polynomials, 
  we establish a one-to-one mapping
  between the square-free cyclic polynomials
  and the binary necklaces (defined below).
The task is then to count the latter.
 
A binary necklace
  is a nonequivalent binary $0$-$1$ string.
Two strings are equivalent if they differ only by a circular shift.
For example, for $n = 3$, there are $2^3 = 8$ binary strings,
but only four necklaces: 000, 001, 011 and 111,
since 010 and 100 are equivalent to 001,
so are 110 and 101 to 011.
The period of a necklace, or a binary string,
  is the length of the shortest non-repeating sub-sequence,
  e.g., the periods of $1111$, $0101$ and $0001$
        are 1, 2, and 4, respectively.
Obviously, the period $m$ divides $n$;
  and a period-$m$ necklace 
  encompasses $m$ binary strings
  differed by circular shifts,
  e.g., $0101$ represents
  both $0101$ and $1010$.

\begin{table}[h]\footnotesize
\caption{Two ways of counting the $N(4) = 6$ necklaces for $n = 4$. 
}
\begin{center}
\begin{tabularx}{\textwidth}{c | c  c | c | c | c | c | c }
\hline
        & \multicolumn{2}{c|}{$m = 1$} 
        & $m = 2$   
        & $m = 4$   & $2^d$ & $\phi(n/d)$ & $\phi(n/d) 2^d$\\ 
\hline 
$d = 1$ & \hspace{1mm} $\mathbf{0}_{\times4}$ \hspace{1mm} 
        & \hspace{1mm} $\mathbf{1}_{\times4}$ \hspace{1mm} & & & $2 = 2$ & $2$ & 4 \\
\hline
$d = 2$ & $\mathbf{00}_{\times2}$ & $\mathbf{11}_{\times2}$
        & \hspace{1mm} $\mathbf{10}_{\times2} (01_{\times2})$
          \hspace{1mm} & & $2 + 2 =4$ & $1$ & 4 \\
\hline
\multirow{3}{*}{$d = 4$} 
      & & & & \hspace{1mm} $\mathbf{1000} (0100, 0010, 0001)$
              \hspace{1mm} & $2 + 2 + 12$& \\
      & $\mathbf{0000}$ & $\mathbf{1111}$ & $\mathbf{1010} (0101)$
            & $\mathbf{1100} (1001, 0110, 0011)$ & $=16$ & $1$ & 16 \\
      & & & &   $\mathbf{1110} (1101, 1011, 0111)$ & &  \\
\hline
$m \sum_{m|d, d|n} \phi(\frac{n}{d})$
      & \multicolumn{2}{c|}{$1\cdot(2+1+1) = 4$}
      & $2\cdot(1+1) = 4$ & $4\cdot1 = 4$ &  \multicolumn{2}{c|}{} &  $\downarrow$\\
\hline
$N(n)\cdot n$ & \multicolumn{2}{c|}{$2 \cdot 4$} & $1 \cdot 4$ & $3 \cdot 4$ & 
  \multicolumn{2}{c|}{$\rightarrow$} & $6\cdot4 = 24$ \\
\hline
\multicolumn{8}{p{\linewidth}}{
  Bold strings are necklaces; others are their cyclic versions.
}
\\
\multicolumn{8}{p{\linewidth}}{
The subscript of a binary string means the number of repeats;
e.g.,
$\mathbf{0}_{\times4}$ means $\mathbf{0}$ repeated four times, or $\mathbf{0000}$;
$\mathbf{10}_{\times2}$ means $\mathbf{10}$ repeated twice, or $\mathbf{1010}$, etc.
}
\\
\hline
\end{tabularx}
\end{center}
\label{tab:countnecklace}
\end{table}

To compute the number of the binary necklaces $N(n)$,
we construct a sum for the length-$n$ binary strings and count it in two ways.
\reftab{countnecklace} shows the example for the $n=4$ case.
For each divisor $d$ of $n$, we collect 
  all length-$n$ binary strings
  whose periods $m$ divide $d$.
The total is $2^d$,
  for we have enumerated all binary strings whose periods divide $d$.
%
We then weight them by the Euler's totient function $\phi(n/d)$.
Here, $\phi(m)$ gives the number of integers from 1 to $m$
  that are coprime to $m$,
  e.g.,
  $\phi(1) = 1$,
  $\phi(3) = 2$ for 1 and 2,
  $\phi(6) = 2$ for 1 and 5.
We repeat the process over other divisors $d$ of $n$,
  and the resulting sum is $\sum_{d|n} \phi(n/d) 2^d$.
The process for a fixed $d$ is exemplified 
  by a row in \reftab{countnecklace}.

We can count the above sum in another way.
We recall that a period-$m$ necklace always contributes 
  $m$ strings in the above process for a fixed $d$,
  and it does so for all multiples $d$ of $m$.
So the total weighted contribution by this necklace is
\begin{equation}
m \sum_{m|d, d|n} \phi(n/d) = 
 m \sum_{\frac{n}{d}|\frac{n}{m}} \phi(n/d) = m \,  \frac{n}{m}  = n,
\label{eq:mphind}
\end{equation}
where we have used the identity $\sum_{d'|m'} \phi(d') = m'$,
with $d' = n/d$ and $m' = n/m$.
Summing over the necklaces yields $N(n) \, n$. Thus,
$N(n) = (1/n) \sum_{d|n} \phi(n/d) 2^d$.
The process for a fixed necklace is exemplified
  by a column in \reftab{countnecklace}.

Back to our problem, there is
  a correspondence between the binary necklaces
  and the square-free cyclic polynomials.
For a square-free cyclic polynomial,
  we construct a binary string according to its generator:
if it contains $x_k$, the $k$th character from the left is 1,
  otherwise 0.
The resulting string corresponds to a unique necklace;
  the alternative generators give
  the circularly-shifted binary strings.
The mapping is reversible, or one-to-one.
For example, if $n = 3$,
 the square-free polynomials for $000$, $100$, $110$ and $111$ are
  $C_0 = 1$ (generator: 1), 
  $C_1 = x_1 + x_2 + x_3$ (generator: $x_1$),
  $C_{12} = x_1 x_2 + x_2 x_3 + x_3 x_1$ (generator: $x_1 x_2$)
  and $C_{123} = x_1 x_2 x_3$ (generator: $x_1 x_2 x_3$), respectively.
\reftab{sqrfreepoly} shows a few more examples.
Thus,

\begin{theorem}
The number $\NB(n)$ of the square-free cyclic polynomials $C_p(\vx)$ ($p \in \Bset$)
  formed by $\vx = \{x_1, \ldots, x_n\}$ is
\begin{equation}
  \NB(n) = N(n) = (1/n) \sum_{d|n} \phi(n/d) 2^{\,d},
\label{eq:necklace}
\end{equation}
which is also equal to $\deg_\lambda A_n\RX$.
\label{thm:necklace}
\end{theorem}
$N(n)$ = 2, 3, 4, 6, 8, 14, 20, 36, 60, 108, 188, 352, 632, \ldots, starting from $n = 1$.

\subsubsection{Number of the $n$-cycles}

\begin{theorem}
The degree in $\lambda$ of the minimal polynomial $P_n\RX$ of all $n$-cycles
  is equal to the number of the $n$-cycles,
  and is given by
\begin{equation}
  \deg_\lambda P_n\RX = L(n) = (1/n)\sum_{d | n} \mu(n/d) 2^{\,d}.
\label{eq:lyndon}
\end{equation}
\label{thm:lyndon}
\end{theorem}

\begin{proof}
Except a few special values of $R$,
the iterated map $f^{n}$ generally
has $2^n$ distinct complex fixed points,
for otherwise $f^{n}(x) - x$ would have a repeated zero at any $R$,
  but at $R=0$, $-x^{2^n} - x$ has no repeated root; a contradiction.

Each fixed point can be assigned to
  a point in a $d$-cycle, with $d$ being a divisor of $n$.
The assignment
  is both complete (for a $d$-cycle point
  must also be a fixed point of $f^n$)
  and non-redundant (for there is no
  repeated fixed point of $f^n$).
Since each of the $L(d)$ $d$-cycles
  contributes $d$ fixed points,
we have
$2^n = \sum_{d\,|\,n} L(d) \, d$.
The M\"obius inversion yields $L(n) = (1/n) \sum_{d\,|\,n} \mu(n/d) 2^{\,d}$.
This formula was known to several authors \cite{hao, lutzky}.

We now define $\lambda_c$ as the value of $\Lambda(\vx)$,
  evaluated at the cycle points
  $\vx^{(c)} \equiv \{ {x_1}^{(c)}, \ldots, {x_n}^{(c)} \}$
  of cycle $c$.
Of course, $\lambda_c$ is a function of $R$.
If we assume that $\lambda_c$ are distinct (see Remark 1 below),
  then the minimal polynomial $P_n\RX$,
  as a polynomial of $\lambda$,
  takes the form of $\prod_c (\lambda - \lambda_c)$.
Thus, the degree of $P_n\RX$ in $\lambda$ 
  must be the same as the number of the $n$-cycles.
\end{proof}

\begin{remark}[1]
Although all trajectory points $\vx^{(c)}$ are distinct in 
  different cycles,
the value $\lambda_c$ of the polynomial $\Lambda(\vx^{(c)})$ 
  may happen to be the same.
In this case,
  we shall find another cyclic polynomial $Y(\vx)$ that
  has different values in the two cycles
  [$Y(\vx)$ exists for otherwise $\vx$ are the same in the two cycles],
  use $\Lambda'(\vx) = \Lambda(\vx) + \epsilon Y(\vx)$
  to list \refeqs{xcp0},
  then take the limit
  $\epsilon \rightarrow 0$.
\end{remark}

\begin{remark}[2]
$L(n)$ is also the number of aperiodic (i.e., period equal to $n$)
  necklaces of length $n$,
  e.g., for $n = 4$, 
  out of the six necklaces, $0001$, $0011$, and $0111$ are aperiodic;
  but $0000$, $1111$, and $0101$ are periodic.
Since a period-$d$ ($d|n$) necklace is just an aperiodic one
  of length-$d$ repeated $n/d$ times,
and each contributes $d$ binary strings,
we can count the $2^n$ binary strings of length $n$ as
  $2^n = \sum_{d|n} L(d) \, d$.
The M\"obius inversion leads to the same result of \refeq{lyndon}.
\end{remark}

$L(n)$ =
2, 1, 2, 3, 6, 9, 18, 30, 56, 99, 186, 335, 630, \ldots, 
starting from $n = 1$.

Since the period $d$ of a length-$n$ necklace always divides $n$,
we have
\begin{equation}
N(n) = \sum_{d|n} L(d).
\label{eq:necklacelyndon}
\end{equation}
We can also show this by explicit computation:
\begin{align*}
\sum_{d | n} L(d)
  &= \sum_{d | n} (1/d) \sum_{c |  d} \mu(d/c) 2^{\,c} 
   = (1/n) \sum_{c | n} 2^{\,c} \sum_{d'  |  (n/c)} 
        \mu\left(\frac{n/c}{d'}\right) d' \\
  &= (1/n) \sum_{c \, | \, n} 2^{\,c} \, \phi(n/c)
  = N(n),
\end{align*}
where we have used $\phi(m) = \sum_{d' | m} \mu(\frac{m}{d'}) \, d'$,
which is the inversion of $m = \sum_{d' | m} \phi(d')$,
and \refeq{necklace}.
We will use \refeq{necklacelyndon} in proving the next theorem.

\subsubsection{Relation between $P_n\RX$ and $A_n\RX$}

\begin{theorem}
The minimal polynomials $P_d\RX$ of all $d$-cycles
  of periods $d|n$
and $A_n\RX$ [defined in \refeq{xr}] are related by
  \begin{equation}
    A_n\RX = \prod_{cd = n} Q_{d, c}\RX,
    \label{eq:prod}
  \end{equation}
  where
  $Q_{d,c}\RX = \prod_{k=1}^{c} P_d(R, e^{2 k \pi i/c} \lambda^{1/c})$
  is a polynomial of degree $L(d)$ in $\lambda$
  representing contributions from $d$-cycles.
\label{thm:prod}
\end{theorem}

\begin{proof}
Since $A_n\RX$ represent all $d$-cycles with $d|n$,
and each cycle holds a distinct $\lambda$,
  $\deg_\lambda A_n\RX$
  is at least $\sum_{d|n} L(d)$
  according to \refeq{lyndon},
  which is equal to
  $N(n)$ by \refeq{necklacelyndon}.
Since $\deg_\lambda A_n\RX = N(n)$ by \refeq{necklace},
  each $n$-cycle occurs exactly once in $A_n\RX$.

In a $d$-cycle,
 we have $\Lambda(\{x_1, \ldots, x_n\}) = (-2)^n x_1 \dots x_n
   = {\Lambda_d}^c$,
where
$\Lambda_d(\{x_1, \ldots, x_d\}) \equiv (-2)^d x_1 \dots x_d$
and
$c \equiv n/d$.
So the $d$-cycle satisfies a polynomial
  $P_d(R, \lambda^{1/c}) = 0$,
where $\lambda = \Lambda(\{x_1,\ldots,x_n\})$.
This is, however, not a polynomial equation, and
the radical $\lambda^{1/c}$ can be removed by the product 
$Q_{d,c}\RX 
  \equiv \prod_{k=1}^c P_d(R, e^{2k\pi i/c} \, \lambda^{1/c}) = 0$.
Now $Q_{d,c}\RX$ is a polynomial of $\lambda$
  for it is invariant under
  $\lambda \rightarrow e^{2\pi i} \lambda$,
and thus free from radicals of the form $\lambda^{l/c}$ 
  \big[if $(l, c) \ne c$\big].
And since 
  $\deg_\lambda Q_{d,c}\RX 
    = \deg_\lambda P_d\RX = L(d)$,
  it is also a polynomial 
  of the lowest possible degree in $\lambda$.

Therefore, the product $\prod_{cd=n} Q_{d,c}\RX$
  can differ from $A_n\RX$ only by
  a multiple.
Since $Q_{n,1}\RX = P_n\RX$,
  and the coefficient of highest power of $\lambda$ is always
  unity in $A_n\RX$ \big[see the definition \refeq{xr}\big],
  we know by induction that the coefficients of the highest power of $\lambda$
  in all $P_n\RX$ and $Q_{d,c}\RX$
  are also unities.
So the multiple is one, hence \refeq{prod}.
\end{proof}

We can now prove \refthm{primfac} as a corollary of \refthm{prod},
\begin{align*}
 B_{n, m} \RX
 & = \prod_{l = 1}^{m} A_n(R, e^{2l\pi i/m} \lambda^{1/m}) \\
 & = \prod_{c | n} \prod_{l=1}^m \prod_{k=1}^c P_{n/c}(R, e^{2k\pi i/c + 2l\pi i/(m c)} \lambda^{1/(m c)}) \\
 & = \prod_{c | n} \prod_{k'=1}^{m c} P_{n/c}(R, e^{2k'\pi i/(m c)} \lambda^{1/(m c)})
  = \prod_{c | n} Q_{n /c, m c}\RX.
\end{align*}
Taking the logarithm (formally) yields
  $\log B_{n, m} = \sum_{d|n} \log Q_{d, m n/d}$,
where $d = n/c$.
The inversion is
  $\log Q_{n, m} = \sum_{d|n} \mu(n/d) \log B_{d, m n/d}$,
or
\begin{equation}
  Q_{n, m}\RX
= \prod_{d|n} B_{d, m n/d}\RX^{\mu(n/d)},
\tag{$\ref{eq:primfac}'$}
\end{equation}
which is reduced to \refeq{primfac} with $m = 1$
for $Q_{n, 1}\RX = P_n\RX$.

\subsection{\label{sec:origfac}Intersection of cycles
  and further factorization at the onset point}

\reftab{Pnlog} shows that
the onset polynomial 
  $P_n(R, \lambda = 1)$ for the $n$-cycles
  can be further factorized.
This is because the intersection of an $n$-cycle
  and a shorter $d$-cycle ($d|n$, $d < n$)
  forces the two to share orbits 
  (this cannot happen if $d \centernot| n$, 
  for the orbits would be out of phase).
Consequently, upon the intersection,
  $P_n\RX$ from the $n$-cycle
  has to accommodate $P_d(R, \lambda')$ from the $d$-cycle,
  with $\lambda'$ being a primitive $(n/d)\,$th root of $\lambda$.

At the intersection,
  the shorter $d$-cycle is branched or ``bifurcated''
  by $(n/d)$-fold to the $n$-cycle.
The simplest example is the first bifurcation point 
  at $R = 3/4$ for $d = 1, n = 2$,
  where the fixed point \refeq{xr1} 
  bifurcates to 
  the 2-cycle \refeq{xr2}.
The second bifurcation point at $R=5/4$ for $d = 2, n = 4$ is similar,
  cf. \refsec{examples}.

We will show below that such branching generally can only happen 
  at the onset of the $n$-cycle, where $\lambda = 1$.
Further, with a real $R$, only a two-fold branching is possible,
  but a complex $R$ allows higher-fold branchings.

If an $n$-cycle is not born out of the above branching,
  we call it an \emph{original} cycle,
  e.g., the 4-cycle at $R_a' = 5/4$
  is born out of bifurcation, see \reffig{cobweb}(d),
  but the other at $R_a = (3+\sqrt[3]{108})/4$
  is original, see \reffig{cobweb}(e),
  also the discussion after \refeq{xr4}.
Both types of cycles exist in $P_n(R, \lambda = +1)$,
  as separate factors;
and the factor responsible for the original cycles, 
  or the \emph{original factor} below, 
  is given by the following formula.

\begin{theorem}
  The original factor $S_n(R)$ of $P_n\RX$ at the onset is given by
  \begin{equation}
    S_n(R)
    = \frac
    {
      P_n(R, 1)
    }
    {
      \prod_{c d =  n, \; c > 1}
      \left[ \,
        \prod_{(k, c) = 1}
      P_{d}
        \left(
          R, e^{2k\pi i/c}
        \right)
      \right]
    },
  \label{eq:origfac}
  \end{equation}
  where the inner product on the denominator is carried over
   $k$ from 1 to $c$ that are coprime to $c$.
  \label{thm:origfac}
\end{theorem}

We illustrate \refthm{origfac} through a few examples
  before giving a proof.
For $n = 1$,
$S_1(R) = P_1(R, 1) = -4R - 1$
as the denominator is unity.

For $n = 2$,
$P_2(R, 1) = 4R - 3$.
But $P_1(R, -1) = -4R + 3$.
So $S_2(R) = P_2(R, 1)/P_1(R, -1) = -1$.
This means that there is no original 2-cycle
  and the only 2-cycle comes from period doubling.

For $n = 3$,
$P_3(R, 1) = -(4R - 7)(16 R^2 - 4 R + 7)$,
whose second factor is equal to
$
\big(-4R + \frac{1 - 3 \sqrt3 i}{2}\big)
\big(-4R + \frac{1 + 3 \sqrt3 i}{2}\big)
= \prod_{k=1,2} P_1(R, e^{2 k \pi i/3})$.
Thus $S_3(R) = -4R + 7$.

For $n = 4$,
$P_4(R, 1)
= (4R - 5) (16 R^2 + 8 R + 5)
  \bigl[
    (4R - 3)^3 - 108
  \bigr]$.
But
$P_2(R, -1) = 4R - 5$ (for $c = 2$)
and
$\prod_{k=1,3} P_1(R, e^{k\pi i/2})
=(-4R-2i-1)(-4R+2i-1)
=16R^2+8R+5$ 
(for $c = 4$).
Dividing $P_4(R, 1)$ by the two factors yields
$S_4(R) = (4R-3)^3 - 108$,
  whose only real root $R=(3+\sqrt[3]{108})/4$ corresponds to
  the onset of the original cycle.
Note $R = 5/4$ is excluded from $S_4(R)$
  as it comes from period-doubling the 2-cycle.


We now prove \refthm{origfac}.
Suppose $n = c \,d$, we have, from \refeq{logmaps},
\[
  x_{l+1} - x_{d+l+1} = - (x_l + x_{d+l}) (x_l - x_{d+l}).
\]
We apply the equation to $l = 1, \ldots, m$, and the product is
\[
  x_{m+1} - x_{m + d+1} =
  (-1)^{m} \left[ \, \prod_{l=1}^{m} (x_l + x_{d+l}) \right]
    (x_1 - x_{d+1}).
\]
We now set $m$ to $0, d, \ldots, (c-1)\,d$ in this equation,
  add them together, eliminate $x_1 - x_{d+1}$
  (which is nonzero in a cycle), and
\begin{equation}
  \sum_{c' = 0}^{c-1}
    (-1)^{c' d} \prod_{l=1}^{c' d} (x_l + x_{d+l})
   = 0.
\label{eq:stair}
\end{equation}
Note \refeq{stair} holds for every divisor $c$ of $n$ ($c > 1$). We now have

\begin{theorem}
An $n$-cycle and a shorter $d$-cycle ($d|n$, $d< n$)
  intersect only at the onset of the $n$-cycle,
and $\prod_{k=1}^{d} f'(x_k) = (-2)^d x_1 \dots x_d$ is
a \emph{primitive} $(n/d)$th root of unity there.
\label{thm:cbifur}
\end{theorem}

\begin{proof}
At the intersection of the $n$- and $d$-cycles, 
  $x_l$ repeats itself after $d$ steps, so
  $x_{d + l} = x_l$;
and \refeq{stair} becomes,
\begin{equation}
  1 + q + \dots + q^{c-1} = 0,
\label{eq:cbifur}
\end{equation}
where $q = (-2)^d \, x_1 \dots x_d$.
Multiplying \refeq{cbifur} by $q - 1$ yields
  $1 = q^c = (-2)^n \, x_1 \dots x_n$.
So the $n$-cycle is at its onset.

Further $q$ is a \emph{primitive} $c$th root of unity.
Suppose the contrary: $q = e^{2k\pi i/c}$ and $(k, c) = g >1$, 
  then by $c_1 \equiv c/g$, $k_1 \equiv k/g$, we have
\begin{equation}
  q = e^{2k_1\pi i/ c_1}.
\label{eq:qd}
\end{equation}
Similar to \refeq{cbifur}, we can apply \refeq{stair} with $c \rightarrow g$ and $d \rightarrow d c_1$, and
\[
  1 + q_1 + \dots + {q_1}^{g-1} = 0,
\]
  where $q_1 = (-2)^{d c_1} \, x_1 \dots x_{dc_1} = q^{c_1}$.
But by \refeq{qd}, $q^{c_1} = e^{2 k_1\pi i} = 1$, 
and $1 + q_1 + \dots + {q_1}^{g-1} = g > 0$; a contradiction.
\end{proof}

\begin{remark}
The only real $q$ is $q = -1$ for $c=2$, i.e., a period-doubling.
On the complex domain, however, we can have a $c$-fold branching with $c>2$,
  which corresponds to a contact points between ``bulbs'' 
  in the Mandelbrot set, see \reffig{logbulb}.
\end{remark}

By \refthm{cbifur}, $P_n\RX$ at the onset point
includes $P_d(R, e^{2 k \pi i/c} \lambda^{1/c})$ for every 
  possible combination of $k$ and $c$, such that 
  $(k, c) = 1$, $c|n$, and $c > 1$.
Dividing the factors from $P_n\RX$ yields \refthm{origfac}.

\subsection{\label{sec:degR}Degrees in $R$}

\begin{theorem}
The degrees in $R$ of $A_n\RX$, $P_n\RX$ and $S_n(R)$ are
\begin{subequations}
\begin{align}
\deg_R A_n\RX &= \sum_{d|n} \phi(n/d) 2^{d-1}, \\
\deg_R P_n\RX &= \sum_{d|n} \mu(n/d) 2^{d-1} \equiv \beta(n), \\
\deg S_n(R) &= \beta(n) - \sum_{d | n, d < n} \beta(d) \, \phi(n/d).
\end{align}
\label{eq:degR}
\end{subequations}
\label{thm:degR}
\end{theorem}

\begin{proof}
We first prove \refeqsub{degR}{a}.
We recall the subscript $p$ of $C_p$ denotes a sequence of indices $k$
  in the generating monomial $\prod_k {x_k}^{e_k}$.
But for a square-free cyclic polynomial,
  each $k$ occurs no more than once,
so $p$ also represents a \emph{set} of indices,
e.g.,
$p = 1$ represents $\{1\}$,
  ($C_1 = x_1 + \dots + x_n$, generator: $x_1$) 
and
$p = 13$ represents $\{1, 3\}$
  ($C_{13} = x_1 \, x_3 + x_2 \, x_4 + \dots + x_n \, x_2$,
  generator: $x_1 \, x_3 $, assuming $n \ge 4$);
more examples are listed in \reftab{sqrfreepoly}).
In this proof,
we shall also use $p$ to denote the corresponding index set,
  $|p|$ the set size, i.e., the number of indices in the set,
  and
  $\bar p \equiv \{1,\ldots,n\} \backslash p$ the complementary set.
Obviously, $|\bar p| + |p| = n$.
Further, we will include $p$ that correspond to alternative generators
  of the same cyclic polynomial,
e.g., we allow $p = \{2\}$, $\{3\}$, \dots, $\{n\}$,
  although they represent the same cyclic polynomial $C_1$
  as $p = \{1\}$.

Next, we recall the matrix elements $T_{pq}(R)$ arise from the square-free reduction
of $\Lambda(\vx) C_p(\vx) = (-2)^n x_1 \dots x_n \, C_p(\vx)$. 
A single replacement ${x_k}^2 \rightarrow R - x_{k+1}$
produces two new terms:
in the first, ${x_k}^2 \rightarrow R$,
and in the second, ${x_k}^2 \rightarrow - x_{k+1}$.
We call the two type 1 and type 2 replacements, respectively.
If a monomial term $t(R, \vx)$ 
results from $l_1$ type 1 and $l_2$ type 2 replacements
during the reduction of a term $s(\vx)$ in $\Lambda(\vx) C_p(\vx)$,
then the degrees in $\vct x$, for any $x_k$, of $s(\vx)$ and $t(R, \vx)$ 
are related as
\begin{equation}
  \deg_\vx s(\vx) - \deg_\vx t(R, \vx) = 2 l_1 + l_2.
\label{eq:degxdiffst}
\end{equation}
Similarly, the degrees in $R$ satisfy
\[
  \deg_R s(\vx) - \deg_R t(R, \vx) = - l_1,
\]
but since $\deg_R s(\vx) = 0$,
\begin{equation}
  \deg_R t(R, \vx) = l_1.
\label{eq:degRt}
\end{equation}

Now if the monomial $t(R, \vx)$ settles in the $q$th column of 
  the matrix $\vct T(R)$, as part 
  of $T_{pq}(R) C_q(\vx)$ in \refeq{xcp0},
then 
  $t(R, \vx)$
  must be 
  a generator of $C_q(\vx)$; 
so
\begin{equation}
  \deg_\vx t(R, \vx) = |q|. 
\label{eq:degxt}
\end{equation}
Since $s(\vx)$ is part of $\Lambda(\vx) C_p(\vx)$,
  we have
\begin{equation}
  \deg_\vx s(\vx) = \deg_\vx \Lambda(\vx) + \deg_\vx C_p(\vx)
  = n + |p|.
\label{eq:degxs}
\end{equation}
From \refeqs{degxdiffst}, \req{degxt}, \req{degxs}, we get
\[
  n + |p| - |q| = 2 \, l_1 + l_2,
\]
and
\begin{equation}
  l_1  =    (n + |p| - |q| - l_2)/2 
             \le  (n + |p| - |q|)/2.
\label{eq:l1limit}
\end{equation}
By \refeq{degRt}, we get
\begin{align*}
  \deg_R T_{pq}(R)   
  = \max\{ \deg_R t(R, \vx) \}
  = \max\{ l_1 \} 
  \le   (n + |p| - |q|)/2,
\end{align*}
where the equality holds when all replacements are type 1 ($l_2 = 0$).

Finally, each term of the determinant 
$A_n\RX = \big|\lambda \, \vct I - \vct T(R)\big|$
is given by 
$(-1)^s \prod_{p} \big[\lambda \, \delta_{p q} - T_{pq}(R)\big]$,
where $p$ runs through rows of the matrix
  and $\{q\}$ is a permutation of $\{p\}$,
  with $(-1)^s$ being the proper sign.
Summing over rows under this condition yields
\begin{align*}
  \deg_R A_n\RX 
  = \max \sum_{p \in \Bset} \deg T_{pq}(R)
  \le N(n) \, \big(n + |p| - |q|\big)/2 = n N(n)/2,
\end{align*}
where equality can be achieved if $q = \bar p$ in every row.
By \refeq{necklace} we have \refeqsub{degR}{a}.
The first few values are 1, 3, 6, 12, 20, 42, 70, 144, 270, 540, 1034, 2112, 4108, \dots, starting from $n = 1$.

To show \refeqsub{degR}{b}, we take the degree in $R$ of \refeq{prod}.
So $\sum_{d|n} \beta(d) \, \big(n/d\big) = N(n) \, n/2$,
whose inversion is
  $\beta(n)/n = \sum_{d|n} \mu(n/d) \, N(d)/2 = L(n)/2$.
The last step follows from inverting \refeq{necklacelyndon}.
The first few values are 1, 1, 3, 6, 15, 27, 63, 120, 252, 495, 1023, 2010, 4095, \dots, starting from $n = 1$.

\refeqsub{degR}{c} follows directly from taking the degree in $R$ of \refeq{origfac}.
The first few values are 1, 0, 1, 3, 11, 20, 57, 108, 240, 472, 1013, 1959, 4083, \dots, starting from $n = 1$.
\end{proof}

\refeqssub{degR}{b} was long known \cite{mira},
and \reqsub{degR}{c} was recently derived \cite{blackhurst}.

\section{\label{sec:henon}H\'enon map}

We now extend the method to the H\'enon map \cite{henon}: 
\begin{equation}
  x_{k+1} = 1 + y_k - a \, {x_k}^2, \quad
  y_{k+1} = b \, x_k.
\label{eq:henon}
\end{equation}
We change variable $x_k \leftarrow a x_k$, $y_k \leftarrow a y_k$, and
\begin{equation}
  x_{k+1} = a + y_k - {x_k}^2, \quad
  y_{k+1} = b \, x_k.
\tag{$\ref{eq:henon}'$}
\label{eq:henons}
\end{equation}
Since neither $a$ nor $b$ is changed during the transformation,
\refeq{henon} and \refeq{henons} share the
same onset and bifurcation points
in terms of $a$ and $b$.
We also see that if $b \rightarrow 0$ and $a \rightarrow R$,
\refeq{henons} is reduced to the logistic map \refeq{logmaps}.

Since $y_k = b \, x_{k-1}$, we can ignore $y_k$ and work with
  cyclic polynomials of $x_k$ only,
the square free reduction is now
${x_k}^2 \rightarrow a + b x_{k-1} - x_{k+1}$.

The stability of \refeq{henons} can be found from the Jacobian matrix
\[
  J_b(x_k)
  \equiv \left(
    \begin{array}{ccc}
      \partial x_{k+1}/\partial x_k & \partial x_{k+1}/\partial y_k \\
      \partial y_{k+1}/\partial x_k & \partial y_{k+1}/\partial y_k \\
    \end{array}
  \right)
  =
  \left(
    \begin{array}{ccc}
      -2 x_k & 1 \\
      b & 0
    \end{array}
  \right).
\]
The eigenvalue $\lambda$ of the composite Jacobian
$J_b(x_1) \cdots J_b(x_n)$ can be computed from
\begin{equation}
\big|\, \lambda \, \vct I - J_b(x_1) \cdots J_b(x_n) \,\big| 
  = \lambda^2 - \Theta(\vx) \,\lambda +(-b)^n = 0,
\label{eq:henonder}
\end{equation}
where $\Theta(\vx)$ and $(-b)^n$
are the trace and determinant of 
the matrix product $J_b(x_1) \cdots J_b(x_n)$,
respectively \cite{hitzl}.
In a stable cycle, the magnitude of $\lambda$
cannot exceed 1; so we replace $\lambda$ by $+1$ or $-1$
in \refeq{henonder}
to obtain the onset or bifurcation point, respectively.
\refeq{henonder} is the counterpart of \refeq{der}.

Since cyclically rotating matrices in a product does not alter the trace,
$\Theta(\vx)$ is a cyclic polynomial of $\vx = \{ x_k \}$.
%
Thus, we can use $\Theta(\vx)$ to list \refeqs{xcp}
and then replace $\Theta(\vx)$ by $\lambda + (-b)^n/\lambda$
or $\pm\big[ 1 + (-b)^n \big]$
in \refeq{xr} to complete the solution.

\newcommand{\BB}[2]{B^{(2)}_{#1,#2}}
\newcommand{\BBB}[3]{B^{(4)}_{#1,#2,#3}}
\newcommand{\BBBB}[4]{B^{(6)}_{#1,#2,#3,#4}}

\begin{table}[h]\footnotesize
\caption{Onset and bifurcation polynomials of the $n$-cycles of the H\'enon map.}
\begin{center}
\begin{tabular*}{\linewidth}{c c c c}
\hline
$n$
& $\Theta(\vx)$
& Onset $P_n(a, b, +1)$ $^{\dagger, \ddagger}$ 
& Bifurcation $P_n(a, b, -1)$ $^{\dagger, \ddagger}$ \\ 
\hline 
1
& $-2 x_1$
& $-A - (b-1)^2$
& $-A + 3(b-1)^2$ \\
2
& $4 C_{12} + 2 b$
& $A-3 (b-1)^2$
& $A - \BB{5}{-6}$ \\
3
& 
  $-8C_{123}-2b C_1$
& 
  \begin{minipage}[h]{.35\linewidth}
  $\left( -A + \BB{7}{10} \right)$ \\
  $\left\{ \left[ A -\frac{1}{2}\BB{1}{-8} \right]^2 + \frac{27}{4} (b^2-1)^2 \right\}$
  \end{minipage}
  
& 
  \begin{minipage}[h]{.35\linewidth}
  \vspace{5mm}
  $\begin{aligned}
  -A^3&+2 \BB{4}{1} A^2 - 9 \BBB{2}{-6}{-7} A \\
  &+9 \BBBB{9}{6}{2}{-10}
  \end{aligned}$
  \vspace{5mm}
  \end{minipage} \\
4
& \begin{minipage}[h]{.2\linewidth}
  $16 C_{1234} + 4 b C_{12} + 2b^2$
  \end{minipage}
& \begin{minipage}[h]{.35\linewidth}
  $\left(A - \BB{5}{-6}\right)$ \\
  $\left\{ \left[A + (b+1)^2\right]^2 + 4(b^2-1)^2 \right\}$ \\
  $\left\{ \left[A - 3(b+1)^2\right]^3 - 108(b-1)^2(b+1)^4 \right\}$
  \end{minipage}
& 
  $\begin{aligned}
  &A^6 - 4A^5 \BB{3}{2} + A^4 \BBB{47}{68}{38} \\
  &-3A^3\BBBB{47}{62}{-83}{-212} + \left({\BBB{17}{12}{-6}}\right)^3  \\
  &+ A^2\BBB{17}{12}{-6}\BBB{31}{-60}{-186}
  \end{aligned}$ 
\\
5
& 
  $\begin{aligned}
  -&32C_{12345}\\
  -&8bC_{123} -2b^2C_1
  \end{aligned}$
&
  \begin{minipage}[m]{.35\linewidth}
  \vspace{3mm}
  $\begin{aligned}
    &-\Big( A^4  -\BB{1}{-12} A^3 + \BBB{1}{6}{21}A^2 \\
    &+\BBBB{9}{-14}{70}{-80} A + \BBB{1}{1}{1}\BBB{31}{-89}{121}
    \Big) \\
    &\Big( A^{11} - \BB{31}{14} A^{10} + 2 \BBB{208}{206}{377} A^9 \\
    &-\cdots -75728722 b - 28629151\Big)
  \end{aligned}$
  \end{minipage}
& 
  $\begin{aligned}
  -A^{15} &+ 2 \BB{16}{1} A^{14}- \BBB{448}{60}{669} A^{13} \\
  &+ \cdots + 1291467969
  \end{aligned}$
\\
\vdots & \vdots & \vdots & \vdots 
\\
9
& 
  $\begin{aligned}
  &-512 C_{1\dots9}\\
  &-\cdots-2b^4 C_1
  \end{aligned}$

& 
  $\begin{aligned}
  &-(A^6 + 12 b A^5 + \dots) 
   \big( A^6 - 4\BB{4}{1}A^5 + \dots \big) \\ 
  &\Big( A^{240} - 8\BB{61}{1}A^{239} +4\BBB{29525}{1019}{58188}A^{238} \\
  &  -\dots +120670698649\dots712084645033 \Big)
  \end{aligned}$

& 
  $\begin{aligned}
  &-\big(
   A^{252} - 504 (1+b^2) A^{251} \\
  +& 4 \BBB{31500}{16}{62067} A^{250} - \dots \\
  +& 5842146539\dots9260477441
  \big)
  \end{aligned}$
\\
\hline
\multicolumn{4}{p{\linewidth}}{
$^\dagger$
Definitions: $A \equiv 4 a$,
$\BB{p}{q} \equiv p (b^2+1) + q b$, 
$\BBB{p}{q}{r} \equiv p (b^4+1) + q (b^3 + b) +r b^2$,
$\BBBB{p}{q}{r}{s} \equiv p (b^6 + 1) + q (b^5 + b) +r (b^4 + b^2) + sb^3$,
$\dots$.
}
\\
\multicolumn{4}{p{\linewidth}}{
$^\ddagger$ 
  The onset polynomials for $n$ from 1 to 4,
  and the bifurcation polynomials for $n$ from 1 to 3,
  agree with those in ref. \cite{hitzl}.
}\\
\hline
\end{tabular*}
\end{center}
\label{tab:henonpoly}
\end{table}

We computed the polynomials of $a$ and $b$ 
  at the onset and bifurcation points for $n$ up to 9.
The polynomials of $a$ and $b$
  at the onset and bifurcation points,
  as well as $\Theta(\vx)$, 
  are listed in \reftab{henonpoly} for small $n$
(for larger $n$, see the website in \refsec{end}).

\section{\label{sec:cubic}Cubic map}

We now study the following cubic map \cite{strogatz}
\begin{equation}
  x_{k + 1} = f(x_k) = r \, x_k - {x_k}^3.
\label{eq:cubic}
\end{equation}
Since the new replacement rule 
\begin{equation}
  {x_k}^3 \rightarrow r \, x_k - x_{k+1}
\label{eq:cubreplace}
\end{equation}
no longer eliminates squares,
we must extend the basis set of cyclic polynomials
  from the square-free ones to the cube-free ones,
  in using \refeq{xcp}.
We include in the basis of expansion
  $C_{112} = {x_1}^2 x_2 + {x_2}^2 x_3 + \dots + {x_n}^2 x_1$ ($n\ge3$),
but not
  $C_{1112} = {x_1}^3 x_2 + {x_2}^3 x_3 + \dots + {x_n}^3 x_1$.

%
However, we only need the cube-free cyclic polynomials of even degrees in $\vx$
to solve the problem,
because \refeq{cubic} contains only linear and cubic terms,
a cyclic polynomial with an odd (even) degree in $\vx$
can never be reduced to one with an even (odd) degree
by \refeq{cubreplace}.
%
%
For technical reasons, we will not use polynomials of odd degrees,
because the map allows a symmetric $2n$-cycle:
$x_1, x_2, \ldots, x_n, -x_1, -x_2, \ldots, -x_n$
(see \reffig{oddcycle}),
which makes all odd cyclic polynomials zero,
e.g., $C_1 = x_1 + x_2 + \dots + x_n - x_1 - x_2 - \dots - x_n = 0$.
Thus, the zero determinant condition, similar to that in \refeq{xr},
  would be useless for these cycles,
  if the odd-degree polynomials were used.

\begin{figure}[h]
  \begin{minipage}{\linewidth}
    \begin{center}
        \includegraphics[angle=-90, width=\linewidth]{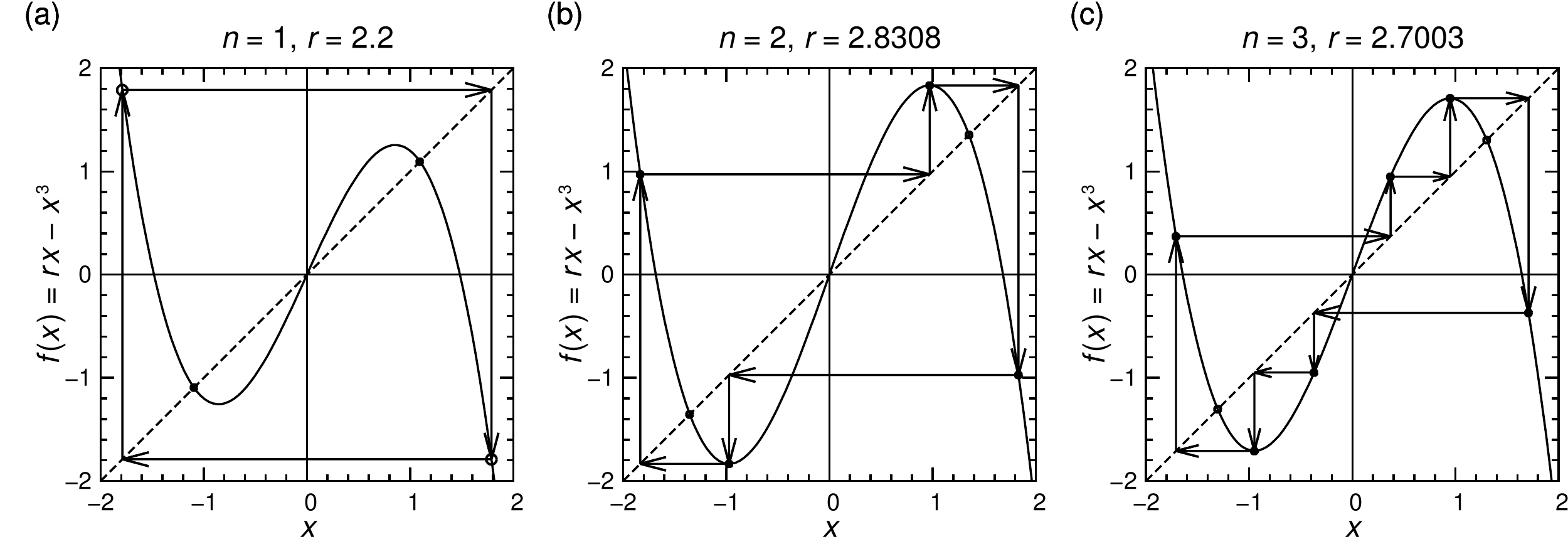}
    \end{center}
  \end{minipage}%
  \caption{\label{fig:oddcycle}
  Odd-cycles of the cubic map.}
\end{figure}

We therefore have a theorem similar to \refthm{sqrfree}.
\begin{theorem}
  For the cubic map \refeq{cubic},
  any cyclic polynomial $K(\vx)$ of
  an $n$-cycle orbit
  $\vct x = \{x_1, \ldots, x_n\}$
  with an even degree in $\vct x$
  is a linear combination of 
the even cube-free cyclic polynomials $C_p(\vx)$:
\[
  K(\vx) = \sum_{p \in \Bset} f_p(R) C_p(\vx),
\]
  where $\Bset = \{0, 11, 12, 13, \ldots, 1122, 1123, \ldots \}$ is
  the set of indices of all even cube-free cyclic polynomials,
  and $f_p(R)$ are polynomials of $R$.
  \label{thm:cubfree}
\end{theorem}

With the above change, the rest derivation is similar to that of the logistic map.
The new $\Lambda(\vx)$ should be 
$\prod_{k=1}^n f'(x_k) = \prod_{k=1}^n (r - 3 {x_k}^2)$.
The polynomials of $r$ 
  at the onset and bifurcation points 
for some small $n$ are shown in \reftab{cubpolygen} (general $\lambda$)
and \reftab{cubpoly} ($\lambda = \pm1$);
for larger $n$ up to 8, 
we have saved the data on the website in \refsec{end}.
The representative $r$ values are listed in \reftab{crval}.
For complex $r$ and $\vx$,
we have plotted \reffig{cubbulb} for regions of stability.

\begin{table}[h]\footnotesize
  \caption{
  Minimal polynomials $P_n(r, \lambda)$ of
  the $n$-cycles of the cubic map \refeq{cubic}.
  }
\begin{center}
\begin{tabular}{lc}
\hline
  $n$ 
& $P_n(r, \; \lambda)$ $^\dagger$  
\\
\hline
1
&
$(\lambda - r) (\lambda + 2 r - 3)$
\\
2
&
  $
  \big[
  \cancel{\lambda - (2r+3)^2}
  \big]
  (\lambda + 2r^2 - 9)
  $
\\
3
&
$\begin{aligned}
  &\lambda^4
  + 2 (r+6)(r^2-9) \lambda^3
  -6 \big(8 r^6+12 r^5-66 r^4-81 r^3+54 r^2-243 r-729\big)
  \lambda^2 
+2 (r^2-9) \big(16 r^7-252 r^5\\
-&216 r^4 +648 r^3+972 r^2+2187 r+4374\big) \lambda 
+2 r (2r-3) (2 r+3)^2 (r^2-9) (2 r^2 - 9)^2 (4 r^2+9) + 531441
\end{aligned}$
\\
4
&
\begin{minipage}{.97\linewidth}
\vspace*{1mm}
$\begin{aligned}
&\big[ \cancel{
  (\lambda-8r^4+54r^2+81)^2
-4 (r^2-9)^2 \lambda
} \big]
\,^\ddagger\,
\big[
\lambda^8
+2(5 r^4-324) \lambda^7
+2(112 r^8-1296 r^6+3807 r^4-91854) \lambda^6 \\
+&\dots
+18075490334784 r^8
-61004779879896 r^6
-411782264189298 r^4
+1853020188851841
\big]
\end{aligned}$ 
\end{minipage}
\\
\vdots & \vdots
\\
\hline
\multicolumn{2}{p{\textwidth}}{
$^\dagger$
The factors from odd-cycles (see \refsec{oddcycle}) are struck out.
}\\
\multicolumn{2}{p{\textwidth}}{
$^\ddagger$
The $n = 4$ odd-cycles satisfy
$(\lambda^\odd + 9)^2 - 2(\lambda^\odd-27) r^2 - 8 r^4 = 0$,
  where $\lambda^\odd=\pm\sqrt\lambda$.
} \\
\hline
\end{tabular}
\end{center}
\label{tab:cubpolygen}
\end{table}

\begin{table}[h]\footnotesize
\caption{Onset and bifurcation polynomials of the $n$-cycles of the cubic map \refeq{cubic}.}
\begin{center}
\begin{tabular*}{\linewidth}{l c c}
\hline
$n$ 
& 
Onset $P_n(r, +1)$ $^\dagger$ 
&
Bifurcation $P_n(r, -1)$ $^\dagger$
\\ 
\hline
1
&
$-(r-1)^2$
&
$-(r+1)(r-2)$
\\
2
&
$-2 (r-2) (r+1) (r+2)^2$
&
$-(r^2-5) (2 r^2+6 r+5)$
\\
3
&
\begin{minipage}{.45\linewidth}
\vspace{3mm} 
$\begin{aligned}
&(r^2+r+1) (4 r^2-14 r+13) \\
&(4 r^8+16 r^7-35 r^6-206 r^5-113 r^4\\
&+376 r^3+715 r^2+1690 r+2197)
\end{aligned}$
\vspace{1mm} 
\end{minipage}
&
\begin{minipage}{.5\linewidth}
\vspace{3mm}
$\begin{aligned}
&16 r^{12}+24 r^{11}-288 r^{10}-434 r^9\\
+&1539 r^8+2358 r^7-1434 r^6-2556 r^5-8541 r^4\\
-&11816 r^3+15288 r^2+24696 r+38416
\end{aligned}$
\vspace{1mm}
\end{minipage}
\\
4
&
\begin{minipage}{.45\linewidth}
\vspace{3mm} 
$\begin{aligned}
-&8(r^2-8) (r^2-5) (r^2+1) (2 r^2-6 r+5) \\
&(2 r^2+6 r+5) (2 r^4-13 r^2-25)^2 \\
&(1024 r^{22}-32512 r^{20}+402304 r^{18} \\
-&2364832 r^{16}+5389924 r^{14}+9715769 r^{12} \\
-&73067038 r^{10}+58934785 r^{8}+235761152 r^6 \\
-&160907264 r^4-671088640 r^2-2097152000)
\end{aligned}$
\vspace{3mm}
\end{minipage}
&
$\begin{aligned}
&-(16 r^8-216 r^6+410 r^4+2142 r^2+1681) \\
&(8192 r^{32}-387072 r^{30}+7834624 r^{28} -88031232 r^{26}\\
&+585876512 r^{24}-2158227720 r^{22} +2211361312 r^{20}\\
&+15958823175 r^{18}  -68871388441 r^{16} +59290039854 r^{14}\\
&+234882618673 r^{12} -524807876277 r^{10} -72612143404 r^8\\
&+308406843576 r^6 +1539579145957 r^4-7984925229121)
\end{aligned}$
\\
\vdots & \vdots & \vdots
\\
8
&
$\begin{aligned}
&-8192 (r^4+1) (2 r^4-18 r^2+41) \\
&(8 r^4-48 r^3+108 r^2-108 r+41) \\
&(8 r^4+48 r^3+108 r^2+108 r+41) \dots\\
&(22016722240\dots14954924752896r^{3108} \\
&-\dots-180097954\dots2522413056\times 10^{400})
\end{aligned}$
&
$\begin{aligned}
-&(17592186044416 r^{80} +\dots+ \\
&144564714832407908402064153121600801) \\
&(4841528421712030\dots03048551060078592 r^{3200}\\
&-\dots-25263420710\dots173884723232001)
\end{aligned}$
\\
\hline
\multicolumn{3}{p{\linewidth}}
{
$^\dagger$
Although $P_n(r, \lambda)$ is the minimal polynomial for a general $\lambda$,
it may contain a pre-factor (e.g., $-2$ in the $n=2, \lambda=+1$ case).
}\\
\hline
\end{tabular*}
\end{center}
\label{tab:cubpoly}
\end{table}

\begin{table}[h]\footnotesize
  \caption{
  Smallest positive $r$ at the onset and bifurcation points
  of the $n$-cycles of the cubic map.
  }
\begin{center}
\begin{tabularx}{\textwidth}{
  >{\hsize=0.5\hsize\centering\arraybackslash}X  
  >{\hsize=1.6\hsize}X  
  >{\hsize=1.6\hsize}X  
  >{\hsize=0.3\hsize\raggedright\arraybackslash}X | 
  >{\hsize=0.5\hsize\centering\arraybackslash}X  
  >{\hsize=1.6\hsize}X  
  >{\hsize=1.6\hsize}X  
  >{\hsize=0.3\hsize\raggedright\arraybackslash}X  
}
\hline
  $n^\dagger$ 
& Onset$^\ddagger$  
& Bifurcation$^\ddagger$
& \#$^*$ 
& 
  $n^\dagger$     
& Onset$^\ddagger$
& Bifurcation$^\ddagger$
& \#$^*$ \\
\hline
$1$     & $1.0000000000_1$      &  $2.0000000000_1$       & 2   & 
$6$     & $2.3334877526_{304}$  &  $2.3355337580_{336}$   & 56  \\
$2'$    & $2.0000000000_1$      &  $2.2360679775_2$       & 2   & 
$6'$    & $2.4608286739_{12}$   &  $2.4657090579_{336}$   & 4   \\
$3$     & $2.4504409645_{8}$    &  $2.4608286739_{12}$    & 4   & 
$7$     & $2.3729872678_{1080}$ &  $2.3732727868_{1092}$  & 156 \\
$4$     & $2.5478350393_{22}$   &  $2.5488312193_{32}$    & 8   &
$8$     & $2.3525990555_{3108}$ &  $2.3527637793_{3200}$  & 400 \\
$4''$   & $2.2360679775_{2}$    &  $2.2880317545_{32}$    & 2   &
$8'$    & $2.5488312193_{32}$   &  $2.5493247379_{3200}$  & 8   \\
$5$     & $2.3939250274_{112}$  &  $2.3957922744_{120}$   & 24  &
$8'''$  & $2.2880317545_{32}$   &  $2.2992279397_{3200}$  & 2   \\
\hline
\multicolumn{8}{p{\textwidth}}{
$^\dagger$
  $\,'$, $\,''$, or $\,'''$ means
    a cycle under
    the first, second, or third successive period-doubling, respectively.
} \\
\multicolumn{8}{l}{
$^\ddagger$
  The subscripts are the degrees of the corresponding minimal polynomial.
} \\
\multicolumn{8}{l}{
$^*$
  The number of similar cycles
    (for $n > 1$, only half of them have positive $r$).
} \\
\hline
\end{tabularx}
\end{center}
\label{tab:crval}
\end{table}

\begin{figure}[h]
  \begin{minipage}{\linewidth}
    \begin{center}
        \includegraphics[angle=-90, width=\linewidth]{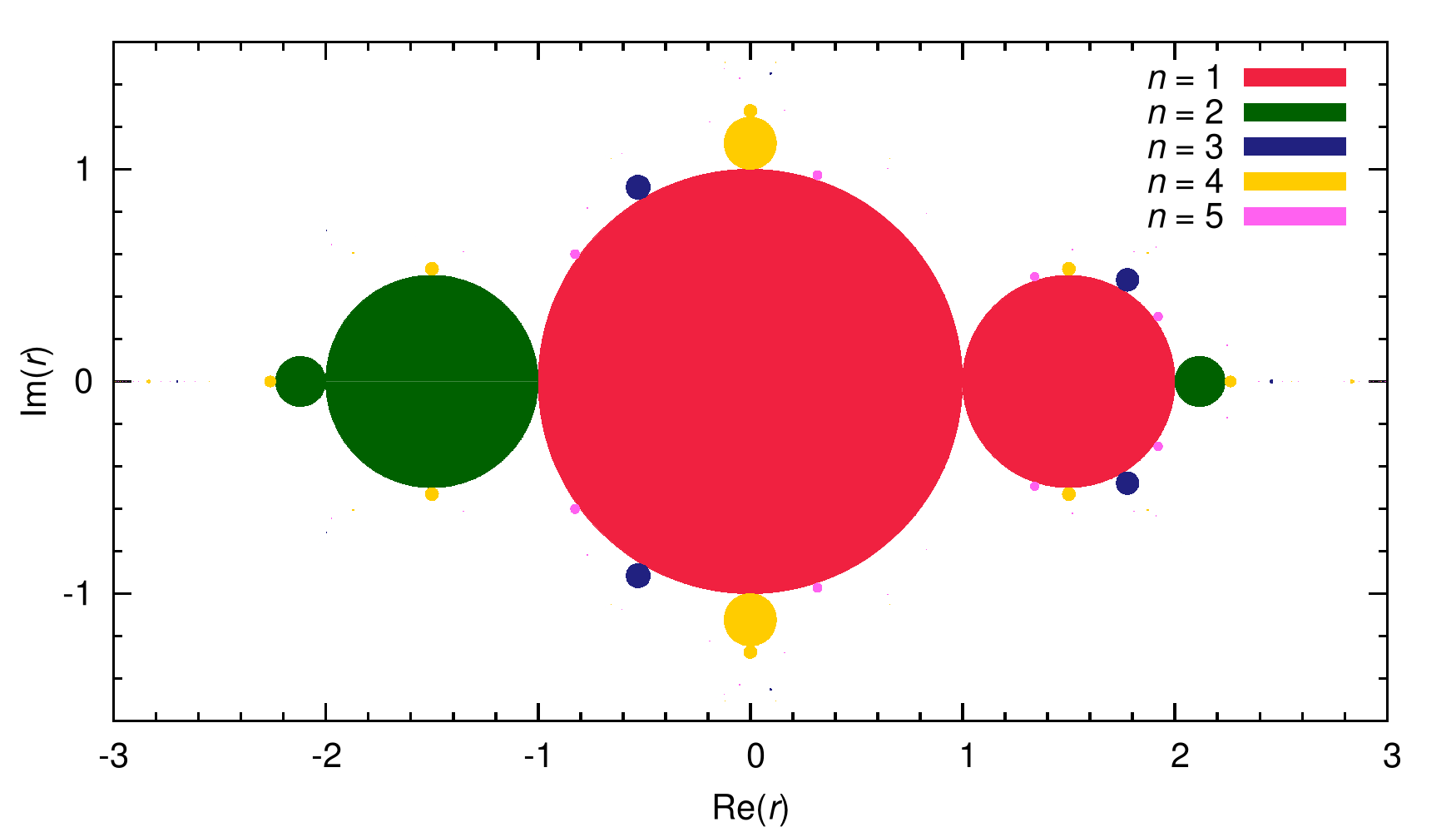}
    \end{center}
  \end{minipage}%
  \caption{
  \label{fig:cubbulb}
  Stable regions of the $n$-cycles of the cubic map \refeq{cubic} 
  with a complex $R$;
  obtained from $P_n(R, e^{i\phi}) = 0$ with $\phi \in [0, 2\pi]$,
  cf. \reffig{logbulb}.
}
\end{figure}

\subsection{Counting cycles}

We now compute the number of the even cube-free cyclic polynomials
by establishing a one-to-one correspondence between
the cube-free cyclic polynomials and the \emph{ternary} necklaces, 
in which each bead of the string is assigned a number
0, 1, or 2, instead of just 0 or 1.
For example, the necklace $\mathbf{212001\cdots}$
corresponds to $C_{112336}(\vx)$, 
whose generator is
${x_1}^2 \, {x_2} \, {x_3}^2 \, x_6$:
the first bead is \textbf{2} for ${x_1}^\mathbf{2}$,
 the second is \textbf{1} for ${x_2}^\mathbf{1}$,
 the third is \textbf{2} for ${x_3}^\mathbf{2}$,
 and the sixth is \textbf{1} for ${x_6}^\mathbf{1}$.
A necklace is even, 
if the corresponding cyclic polynomial has an even degree in $\vx$.
This means that the sum of numbers (0, 1, or 2)
  on the beads of the necklace, 
  which equals the degree in $\vx$ of the polynomial, is also even.

%

\begin{theorem}
The number of the even ternary necklaces 
or the cube-free cyclic polynomials for the cubic map of even degrees in $\vx$ 
is given by
\begin{equation}
  N_e(n) = \frac{1}{n} \sum_{c d = n} \phi(c)
    \left[
      3^d - \odd(c) \frac{3^d-1}{2}
    \right],
\label{eq:cubnecklace}
\end{equation}
where
$\odd(c) \equiv \big[1 - (-1)^c\big]/2$ is 1 if $c$ is odd or 0 if even.
\label{thm:cubnecklace}
\end{theorem}

\begin{proof}

We first show that the number of even ternary strings is $(3^n + 1)/2$.
Consider the generating function 
\[
  Z\big(\{x_1, x_2,\dots, x_n\}\big) = \prod_{i=1}^n (1+x_i+{x_i}^2),
\]
where $1$, $x_i$, and $x_i^2$ correspond to the states that bead $i$
taking the number $0$, $1$, and $2$, respectively;
and the product over $n$ sums over states of independent beads.
In the expansion of $Z(\vx)$,
 each term (which takes the form
  ${x_1}^{e_1} {x_2}^{e_2} \dots {x_n}^{e_n}$,
  with $e_i = 0, 1, 2$)
represents a unique ternary string $e_1 e_2 \dots e_n$,
which is even, if $e_1 + e_2 + \dots + e_n$ is so.
By setting $x_1 = x_2 = \dots = x_n = 1$,
$Z=3^n$ equals the total number of ternary strings.
By setting $x_1 = x_2 = \dots = x_n = -1$,
a term representing an even (odd) ternary string
  is $+1$ ($-1$);
and $Z=(1-1+1)^n=1$ 
equals the difference between the number of even strings
  and that of odd strings.
Thus, the average $(3^n+1)/2$
gives the number of even ternary strings.

\begin{table}[h]\footnotesize
\caption{Two ways of counting the $N_e(2) = 4$ ternary necklaces for $n = 2$. 
}
\begin{center}
\begin{tabularx}{\textwidth}{c | c  c c | c  | c | c | c | c }
\hline
        & \multicolumn{3}{c|}{$m = 1$} 
        & $m = 2$   
        & $c = n/d$
        & $T(d, c)$ $^\dagger$
        & $\phi(c)$ & $\phi(n/d) T(d, c)$\\ 
\hline 
$d = 1$ & \hspace{1mm} $\mathbf{0}_{\times2}$ \hspace{1mm} 
        & \hspace{1mm} $\mathbf{1}_{\times2}$ \hspace{1mm} 
        & \hspace{1mm} $\mathbf{2}_{\times2}$ \hspace{1mm} 
        & & $2$ & $3$ & $1$ & 3 \\
\hline
$d = 2$ & $\mathbf{00}$ 
        & $\mathbf{11}$
        & $\mathbf{22}$
        & \hspace{1mm} $\mathbf{20} (02)$ $^\ddagger$ \hspace{1mm} 
        & $1$ & $3+1\cdot2=5$ & $1$ & 5 \\
\hline
$m \sum_{m|d, d|n} \phi(\frac{n}{d})$
      & \multicolumn{3}{c|}{$1\cdot(1+1) = 2$}
      & $2\cdot1 = 2$ &    \multicolumn{3}{c|}{} &  $\downarrow$\\
\hline
$N_e(n)\cdot n$ 
    & \multicolumn{3}{c|}{$3 \cdot 2$} & $1 \cdot 2$ & 
  \multicolumn{3}{c|}{$\rightarrow$} & $4\cdot2 = 3 + 5$ \\
\hline
\multicolumn{9}{p{\linewidth}}{
  Bold strings are necklaces; others are their cyclic versions.
}
\\
\multicolumn{9}{p{\linewidth}}{
The subscript of a string means the number of repeats;
e.g.,
$\mathbf{1}_{\times2}$ means $\mathbf{1}$ repeated twice, or $\mathbf{11}$;
}
\\
\multicolumn{9}{p{\linewidth}}{
$^\dagger$ $T(d, c) \equiv 3^d - \odd(c)(3^d-1)/2$,
which is $3^d$ if $c$ is even,
or $(3^d+1)/2$ if $c$ is odd.
}
\\
\multicolumn{9}{p{\linewidth}}{
$^\ddagger$
The odd binary strings $\mathbf{10}$, $01$, $\mathbf{21}$, and $12$
  do not contribute to the sum, and are excluded from $T(d, c)$.
}
\\
\hline
\end{tabularx}
\end{center}
\label{tab:countcubnecklace}
\end{table}

The rest counting process is similar to that in \refsec{degprimfac}:
we construct the sum $n N_e(n)$ by two ways, as exemplified in \reftab{countcubnecklace}.
In the first way, 
  for a fixed $d$ ($d|n$),
  if $c\equiv n/d$ is even, 
    we count all ternary strings whose period $m$ divide $d$,
  but if $c$ is odd, we count only ternary strings
    whose first $d$ beads are even
    (because repeating an odd string an odd number times 
    does not yield an even string);
  in either case, we multiple the result by $\phi(c)$.
The process for a fixed $d$ corresponds to a row of \reftab{countcubnecklace}.
Repeating the process for all $d$ gives
  $\sum_{cd = n} T(d, c) \, \phi(c)$,
  where $T(d, c)$ is the total number $3^d$ of period-$d$ ternary strings if $c$ is even,
  or the number $(3^d+1)/2$ of even strings if $c$ is odd.

In the second way,
  we look at the contribution from each necklace to the above sum.
An even period-$m$ necklace contributes a total of
  $m \times \sum_{m|d, d|n} \phi(n/d) = n$
  \big[the multiplier $m$ is for the $m$ cyclic versions, cf. \refeq{mphind}\big],
  while an odd necklace contributes nothing.
Thus, the sum equals $N_e(n)\cdot n$.
The process for a fixed necklace corresponds to 
  a column of \reftab{countcubnecklace}.

So
\[
  N_e(n) \cdot n = \sum_{cd = n} T(d, c) \, \phi(c),
\]
which is \refeq{cublyndon}
after we divide both sides by $n$.
\end{proof}

$N_e(n)$ = 2, 4, 6, 14, 26, 68, 158, 424, \dots, starting from $n = 1$.

The characteristic polynomial $A_n(r, \lambda)$
  from the determinant equation 
  has a degree $N_e(n)$ in $\lambda$.
Again, it encompasses the factors for the $n$-cycles 
  and the shorter $d$-cycles, as long as $d|n$.
The minimal polynomial for the $n$-cycles can be obtained 
  by \refthm{primfac} with proper substitutions
  \big[$R\rightarrow r$, 
    $A_n(R, \lambda) \rightarrow A_n(r, \lambda)$,
    etc.\big].
The degree of the polynomial is given by

\begin{theorem}
The degree in $\lambda$ of 
the minimal polynomial $P_n(r, \lambda)$ of the $n$-cycles is
\begin{equation}
  L_e(n) = \frac{1}{n} \sum_{cd = n} \mu(c)
    \left[
      1 + \odd(c)\frac{3^d-1}{2}
    \right],
\label{eq:cublyndon}
\end{equation}
where $\odd(n) = [1 - (-1)^n]/2$ 
is $1$ for an odd $n$,
but $0$ for an even $n$.
\label{thm:cublyndon}
\end{theorem}

\refeq{cublyndon} follows from the inversion
  $L_e(n) = \sum_{d|n} \mu(n/d) \, N_e(n)$,
after some algebra, as shown in \refapd{cublyndon}.
$L_e(n)$ = 2, 2, 4, 10, 24, 60, 156, 410, \ldots, starting from $n = 1$.
This is also the number of the $n$-cycles \cite{hao}.
Note that, for $n>1$, half of the cycles have negative $r$,
and the $x_k$ are imaginary.
However, in a transformed map,
\[
z_{k+1} = r \, z_k \, (1 - {z_k}^2),
\]
which differs from \refeq{cubic} by $x_k = \sqrt{r} z_k$,
$z_k$ in the negative-$r$ cycles are real.

Following a similar proof to \refthm{degR},
we find the corresponding degrees in $r$
of the characteristic polynomial $A_n(r, \lambda)$
  and minimal polynomial $P_n(r, \lambda)$ of the $n$-cycles
are $n N_e(n)$ and $n L_e(n)$, respectively.

\subsection{\label{sec:oddcycle}Odd-cycles}

Because of the symmetry $f(-x) = -f(x)$,
the minimal polynomial $P_n(r, \lambda)$ 
is subject to factorization for an even $n$.
If $x_{(n/2)+1} = -x_1$,
  then $x_1, \ldots, x_{n/2}, -x_1, \ldots, -x_{n/2}$ 
  is an $n$-cycle, 
  for $x_{n+1} = - x_{(n/2)+1} = x_1$.
We call such a cycle an odd-cycle,
see \reffig{oddcycle} for examples.
Odd-cycles satisfy a polynomial 
  of lower degrees in $\lambda$,
  which causes the factorization.
Suppose $\Lambda^\odd(\vx) \equiv \prod_{k = 1}^{n/2} f'(x_k)$ in the odd-cycle
 satisfies $P_{n/2}^\odd(r, \lambda^\odd) = 0$ 
\big[where $\lambda^\odd$ is the value of $\Lambda^\odd(\vx)$,
and $\lambda^\odd = \pm\sqrt \lambda$\,\big],
then $P^\odd_{n/2}(r, \sqrt \lambda) P^\odd_{n/2}(r, -\sqrt \lambda)$ 
is a factor of $P_{n}(r, \lambda)$.
For example, by solving the $n = 2$ odd-cycle,
  we have $- x_1 = r x_1 - {x_1}^3$, or ${x_1}^2 = r + 1$.
Since $\lambda^\odd = r - 3 {x_1}^2$,
$P^\odd_1(r, \lambda^\odd) = \lambda^\odd + 2 r + 3$.
Now the factor for 2-cycles (see \reftab{cubpolygen})
  is
  $P_2(r,\lambda) = -\big[(2r+3)^2 - \lambda\big]\,(\lambda + 2r^2-9)$,
  whose first factor is indeed 
  $P^\odd_1(r, \sqrt \lambda) \, P^\odd_1(r, -\sqrt \lambda)
   = (\sqrt \lambda + 2 r + 3)(-\sqrt \lambda + 2 r + 3)$.
The example for $n=4$ is shown in \reftab{cubpolygen}.

\section{\label{sec:end}Summary and discussions}

We now summarize the algorithm for a one-dimensional polynomial map.
First, we list \refeqs{xcp}  with $\Lambda(\vx) = \prod_{k=1}^n f'(x_k)$.
This step populates elements of the matrix $\vct T(r)$,
where $r$ is the parameter of the map.
The determinant $A_n(r, \lambda) = \big|\lambda \, \vct I - \vct T(r)\big|$,
with $\lambda$ being $+1$ and $-1$, then gives the characteristic 
polynomial at onset and bifurcation points, respectively.
To filter out factors for the shorter $d$-cycles with $d|n$,
  we repeat the process for other divisors $d$ of $n$
  and then apply \req{primfac}.

When implemented on a computer,
it is often helpful to evaluate $A_n(r,\lambda)$ by Lagrange interpolation,
that is, we evaluate $A_n(r, \lambda)$ at a few different $r$,
e.g., $r = 0, \pm1, \pm2,\ldots$, then piece them together
to a polynomial.
The strategy also allows a trivial parallelization.

The algorithm (implemented as a Mathematica program)
was quite efficient.
For the logistic map, the bifurcation point for $n = 8$
took three seconds to compute on a desktop computer
(single core, Intel\textsuperscript{\textregistered} Dual-Core CPU 2.50GHz).
In comparison,
  the same problem took roughly 5.5 hours \cite{kk1}
  using Gr\"obner basis
  and 44 minutes in a later study \cite{lewis}.
To be fair, using the latest Magma, computing the Gr\"obner basis
  took 81 and 14 minutes, on the same machine
  for \refeq{logmap} and \refeq{logmaps}, respectively;
even so, our approach still had a 200-fold speed-up.

The exact polynomials of these maps are generally too large to print on paper,
  e.g., the polynomial for the logistic map with $n = 13$
    takes roughly seven megabytes to write down.
We therefore save the polynomials and programs of the three maps on the web site
http://logperiod.appspot.com.


\section*{Acknowledgements}

I thank Drs. T. Gilbert and Y. Mei for helpful communications.
Computing time on the Shared University Grid at Rice, 
funded by NSF under Grant EIA-0216467, is gratefully acknowledged.

\appendix

\section{\label{apd:per4}Simple derivation of 4-cycles}

The polynomials for the 4-cycles permits a short derivation.
We first list the explicit equations:
\begin{subequations}
\label{eq:x4}
\begin{align}
  x_2 &= R - x_1^2, \\
  x_3 &= R - x_2^2, \\
  x_4 &= R - x_3^2, \\
  x_1 &= R - x_4^2,
\end{align}
\end{subequations}
$\big[\mathrm{\refeqsub{x4}{a}} - \mathrm{\refeqsub{x4}{c}}\big]
\times
\big[\mathrm{\refeqsub{x4}{b}} - \mathrm{\refeqsub{x4}{d}}\big]$
yields $1 + (x_1 + x_3) (x_2 + x_4) = 0$, 
since $x_1 \ne x_3, x_2 \ne x_4$.
Hence, with
$y_1 \equiv x_1 + x_3$, $y_2 \equiv x_2 + x_4$,
$z \equiv y_1 + y_2$,
we have
\begin{subequations}
\begin{align}
y_1 y_2       &= -1 \\
y_1^2 + y_2^2 & = (y_1 + y_2)^2 - 2 y_1 y_2 = z^2 + 2 \\
y_1^3 + y_2^3 & = (y_1 + y_2)^3 - 3 y_1 y_2 (y_1 + y_2) = z^3 + 3 z.
\end{align}
\label{eq:ypow4}
\end{subequations}
Multiplying \refeqsub{x4}{a} by $x_1$ or $x_3$,
then summing over cyclic versions yields
\begin{subequations}
\begin{align}
y_1 y_2 &= R z - \big[(x_1^3 + x_3^3) + (x_2^3 + x_4^3)\big],\\
y_1 y_2 &= R z - \big[x_1 x_3 (x_1 + x_3) + x_2 x_4 (x_2 + x_4)\big].
\end{align}
\label{eq:p4q}
\end{subequations}
From 
$\mathrm{\refeqsub{p4q}{a}} + 3 \times \mathrm{\refeqsub{p4q}{b}}$,
we have
$4 \, y_1 y_2 = 4 R z - (y_1^3 + y_2^3)$,
and by \refeqs{ypow4},
\begin{equation}
  z^3 - (4 R - 3) z - 4 = 0.
  \label{eq:xr4s}
\end{equation}
Since $2 x_1 x_3 = y_1^2 - (x_1^2 + x_3^2) = y_1^2 - 2 R + y_2$,
and $2 x_2 x_4 = y_2^2 - 2 R + y_1$,
%
\begin{equation}
  X \equiv x_1 x_2 x_3 x_4 = \frac{1}{2} R z(1 - z) + (R^2 - R + 1),
\label{eq:der4}
\end{equation}
where we have used \refeqs{ypow4} and \req{xr4s} to simplify the result.
%
Dividing the polynomial in \refeq{xr4s} by that in \refeq{der4}
yields $z = (R^2-3R-X+1)/(R^2-R+X-1)$,
and plugging it back to \refeq{der4} gives
   $R^6 -3 R^5
  + (3 + X) (R^4 - R^3)
  + (1 - X) (2 + X) R^2
  + (1 - X)^3 = 0$,
which is the same as the first factor of \refeq{xr4}
  with $X = \lambda/16$.

\section{\label{apd:cublyndon}Proof of \refthm{cublyndon}}

Here we prove \refthm{cublyndon} \big[or \refeq{cublyndon}\big]
  for the cubic map.
Similar to the logistic map case \refeq{necklacelyndon},
  we have, for the cubic map,
\[
  N_e(n) = \sum_{d|n} L_e(d).
\]
Thus, we only need to inverse this equation to obtain $L_e(n)$.
But owing to the complexity of \refeq{cubnecklace},
  we need the Dirichlet generating function
  to simplify the result.

For a series $\alpha(n)$,
the Dirichlet generating function is defined as
\[
    G_\alpha(s) \equiv \sum_{n=1}^\infty \alpha(n) \, n^{-s}.
\]
In \reftab{genfunc}, we list the generating functions of 
some common series,
and define a few new ones for $N_e(n)$ and $L_e(n)$, etc..

\begin{table}[h]\footnotesize
  \caption{Dirichlet generating functions for ternary necklaces.}
\begin{center}
\begin{tabularx}{\textwidth}{
  >{\hsize=0.6\hsize\raggedleft\arraybackslash}X  
  >{\hsize=1.4\hsize\centering\arraybackslash}X  
  >{\hsize=0.6\hsize\raggedleft\arraybackslash}X  
  >{\hsize=1.4\hsize\centering\arraybackslash}X 
}
\hline
$\alpha(n)$   &   $G_\alpha(s) = \sum_{n} \alpha(n)/n^{s}$  &
$\alpha(n)$   &   $G_\alpha(s) = \sum_{n} \alpha(n)/n^{s}$ 
\\
\hline
1             & $\zeta(s)$ $^\dagger$ &
\multicolumn{1}{|r}
{$N_e(n)$}      & $G_N(s)$ $^\ddagger$
\\
$\delta_{n,1}$         & $1$ $^\dagger$ &
\multicolumn{1}{|r}
{$n N_e(n)$}  & $G_N(s-1)$ $^\ddagger$
\\
$\mu(n)$    & $\zeta(s)^{-1}$ $^\dagger$ &
\multicolumn{1}{|r}
{$L_e(n)$}    & $G_L(s)$ $^\ddagger$
\\
$\phi(n)$   & $\zeta(s-1)/\zeta(s)$ $^\dagger$ &
\multicolumn{1}{|r}
{$nL_e(n)$}   & $G_L(s-1)$ $^\ddagger$
\\
$\odd(n)$   & $\zeta(s) (1-2^{-s})$ &
\multicolumn{1}{|r}
{$3^n$}       & $t(s)$ $^\ddagger$
\\
\cline{3-4}
$\mu(n) \, \odd(n)$
  & $\Big[ \zeta(s) (1-2^{-s}) \Big]^{-1}$ &
$\phi(n) \, \odd(n)$
  & $\frac{\zeta(s-1)}{\zeta(s)} \frac{1-2^{-s+1}}{1-2^{-s}}$ 
\\
\hline
\multicolumn{4}{p{\linewidth}}{
$^\dagger$
$\zeta(s) = \sum_n n^{-s}$ is the zeta function.
See ref. \cite{hardy} for proofs. 
}\\
\multicolumn{4}{p{\linewidth}}{
$^\ddagger$ 
The sum is truncated at a large $M$ to avoid divergence.
}\\
\hline
\end{tabularx}
\end{center}
\label{tab:genfunc}
\end{table}

The generating function has an important property:
$G_\gamma(s) = G_\alpha(s) G_\beta(s)$, 
if and only if $\gamma(n) = \sum_{d|n} \alpha(n/d)\, \beta(d)$ \cite{hardy}.
Thus,
  the terms of the sum $\sum_{d|n} \alpha(n/d)\, \beta(d)$ 
  of two sequences $\alpha$ and $\beta$
  can be readily found from expanding the generating function.

Another fact is if $G_\alpha(s)$ is the generating function of $\alpha(n)$,
then $G_\alpha(s-1)$ is that of $n \alpha(n)$,
for $\alpha(n)/n^{s-1} = \big[ n \alpha(n) \big] / n^s$.
Thus, the generating function of $n$ is $\zeta(s-1)$,
and
that of $n N_e(n)$ is $G_N(s-1)$ 
\big[$\zeta(s)$ is the generating function of 1,
  and $G_N(s)$ is that of $N_e(n)$,
  see \reftab{genfunc}\big].

We now compute the generating function of $\mu(n)\odd(n)$.
First, recall the generating function $G_\mu(s)$ of $\mu(n)$ is
\[
  G_\mu(s) = \sum_{n=1}^{\infty} \frac{\mu(n)}{n^s} 
  = \prod_p \left( 1 - \frac{1}{p^s} \right),
\]
where $p$ goes through every prime.
The follows directly from expanding the product 
  and the definition of $\mu(n)$,
which is $-1$ to the power of the number of distinct prime factors.
The same reasoning applies to 
$G_{\mu, \odd}(s) = \sum_{n \text{ odd}} \mu(n)/n^s$
with the only difference being that all multiples $n$ of 2
are absent. So
\[
  G_{\mu, \odd}(s) 
  = \sum_{n=1}^{\infty} \frac{\mu(n) \, \odd(n)}{n^s}
  = \sum_{n \text{ odd}} \frac{\mu(n)}{n^s}
  = \prod_{p \ge 3} \left( 1 - \frac{1}{p^s} \right).
\]
Comparing the two formulas yields
\[
  G_{\mu, \odd}(s) 
  = G_{\mu}(s) 
  \left(1- \frac{1}{2^s}\right)^{-1} 
  = 
   \left[ \zeta(s) \left(1- \frac{1}{2^s}\right) \right]^{-1}. 
\]

Similarly,
we can compute the generating function of $\odd(n)$ as
\[
  G_\odd(s) 
  \equiv \sum_{n=1}^{\infty} \frac{\odd(n)}{n^s}
  = \sum_{n \text{ odd}} \frac{1}{n^s}
  = \zeta(s) \left(1 - \frac{1}{2^s} \right).
\]
This can also be derived by taking the generating function of 
both sides of the identity:
$\sum_{d|n} \mu(d) \odd(d) \, \odd(n/d) = \delta_{n, 1}$
\big[which is a modification of
$\sum_{d|n} \mu(d) = \delta_{n, 1}$\big].
It follows that the generating function of $n \, \odd(n)$ is $G_\odd(s-1)$.

The generating function $G_{\phi,\odd}(s)$
of $\phi(n) \odd(n)$ can be computed by taking the generating function
of both sides of the identity
\[
  n \, \odd(n)  = \sum_{d|n} \phi(d) \, \odd(d) \; \odd(n/d),
\]
i.e., if $n$ is even, then both sides are 0; 
if odd, then $n = \sum_{d|n} \phi(d)$.
So
\[
  G_{\phi,\odd}(s) = \frac{G_\odd(s-1)}{G_\odd(s)}
  = \frac{\zeta(s-1)}{\zeta(s)} \frac{1 - 2^{-s + 1}}{1 - 2^{-s}}.
\]

We can now compute the generating function $G_N(s)$ of $N_e(s)$.
By multiplying $n$ to both sides of \refeq{cubnecklace},
and taking the generating function, 
we find that
\begin{align*}
  G_N(s - 1)
    = G_\phi(s) t(s)- G_{\phi, \odd}(s) \frac{t(s)-\zeta(s)}{2}
    = 
      \frac{\zeta(s-1) \Big[ t(s) + \zeta(s) \big( 1  - 2^{-s+1} \big) \Big]  }
           {2 \, \zeta(s) (1 - 2^{-s}) }
\end{align*}
where the left side $G_N(s-1)$ is
  the generating function of $N_e(n) \, n$,
and formulas in \reftab{genfunc} have been used.

Finally, we take the generating function of both sides of
 $N_e(n) = \sum_{d|n} L_e(d)$:
\[  
  G_N(s) = \zeta(s) \, G_L(s),
\]
and
\begin{align*}
  G_L(s-1) = \frac{G_N(s - 1)}{\zeta(s-1)}
       = 1
        + \frac{ t(s) - \zeta(s) }
          {2 \, \zeta(s) \big(1 - 2^{-s}\big)}
       = 1 + \frac{t(s) - \zeta(s)}{2} G_{\mu, \odd}(s). 
\end{align*}
Comparing the coefficients of the $n$th term ($n \ll M$),
  we find 
\[ 
  n L_e(n) = \delta_{n,1} + \sum_{cd = n} \mu(c) \, \odd(c) \frac{3^d-1}{2},
\]
which is \refeq{cublyndon}
\big[also note $\delta_{n,1} = \sum_{c|n}\mu(c)$\big].

\end{document}